\documentclass[journal]{IEEEtran}

\ifCLASSINFOpdf
\else
   \usepackage[dvips]{graphicx}
\fi

\usepackage[T1]{fontenc}
\usepackage[utf8]{inputenc}
\usepackage[english]{babel}
\usepackage{booktabs}
\usepackage{xr}
\usepackage{microtype}
\usepackage{csquotes} 
\usepackage{textcomp} 

\usepackage{siunitx}
\usepackage{algorithm} 
\usepackage{algorithmic}
\usepackage{newunicodechar}
\usepackage{pifont}
\usepackage{soul} 

\usepackage{amsmath}
\usepackage{amsfonts}
\usepackage{amssymb}
\usepackage{amsthm}
\usepackage{amsopn}
\usepackage[cmintegrals]{newtxmath}
\usepackage{bm}
\usepackage{mathtools}

\usepackage{graphicx}
\graphicspath{{figures/}}
\usepackage{xcolor}
\usepackage[skip=2pt]{subcaption}
\DeclareCaptionLabelSeparator{ieeefigsep}{.\nobreakspace\nobreakspace}
\usepackage{tikz}

\usepackage{multirow}
\usepackage{longtable} 
\usepackage{tabu}
\usepackage{booktabs}
\usepackage{tabularx}

\usepackage{calc}
\usepackage{todonotes} 
\usepackage{xifthen} 
\usepackage{xparse} 
\usepackage{etoolbox} 
\usepackage{xstring}
\usepackage{xspace} 

\usepackage{url}
\usepackage{hyperref}
\hypersetup{
	colorlinks=true, 	
	linkcolor=black, 	
	citecolor=black, 	
	filecolor=black, 	
	urlcolor=black 		
}
\usepackage{cleveref}
\crefname{table}{Table}{Tables} 
\crefname{figure}{Fig.}{Fig.} 
\crefname{section}{Section}{Section} 
\crefname{equation}{Equation}{Equation} 
\usepackage{url}
\usepackage{lipsum}

\newcommand*{\norm}[2][]{\left\lVert#2\right\rVert_{#1}}

\newcommand*{\twonorm}[1]{\norm[2]{#1}}

\newcommand*{\fronorm}[1]{\norm[F]{#1}}

\newcommand*{\toeplitz}[1]{T_{#1}}

\newcommand*{\C}{\mathbb{C}}
\newcommand*{\R}{\mathbb{R}}

\newcommand*{\N}{\mathbb{N}}

\newcommand*{\mat}[1]{\bm{#1}}

\renewcommand*{\vec}[1]{\bm{#1}}

\newcommand*{\pinv}{^\dagger}

\newcommand*{\suchthat}{\mid}

\DeclareMathOperator{\rank}{rank}

\newcommand{\ie}{i.e.\@\xspace} 
\newcommand{\etal}{\textit{et al.\@}\xspace}

\newtheorem{myth}{Theorem}
\newtheorem{proposition}{Proposition}

\providecommand{\gnorm}[1]{\left\Vert #1 \right\Vert_{\mat{\Gamma}}}
\providecommand{\ginner}[2]{\left\langle #1, #2\right\rangle_{\mat{\Gamma}}}
\providecommand{\Fix}{\textnormal{Fix}}
\providecommand{\Mod}{\mathcal{M}}
\providecommand{\gnormop}[1]{\left\Vert #1 \right\Vert_{\mat{\Gamma}}}

\makeatletter
\@ifundefined{definition}{\newtheorem{definition}{Definition}}{}
\@ifundefined{lemma}{\newtheorem{lemma}{Lemma}}{}
\@ifundefined{corollary}{\newtheorem{corollary}{Corollary}}{}
\@ifundefined{remark}{\newtheorem{remark}{Remark}}{}
\makeatother

\begin{document}

\title{Cadzow Projected Gradient Descent for Generalized Finite Rate of Innovation: A Quantitative Local Convergence Theory}

\author{A. Besson, \IEEEmembership{Member, IEEE}
\thanks{A.~Besson is with E-Scopics SAS, Aix-en-Provence, France (e-mail: adrien.besson@e-scopics.com).}}

\markboth{IEEE Transactions on Signal Processing}%
{Besson \MakeLowercase{\textit{et al.}}: Cadzow Projected Gradient Descent for Generalized Finite Rate of Innovation: A Quantitative Local Convergence Theory}
\maketitle

\begin{abstract}
The generalized finite rate of innovation~(GenFRI) framework aims at reconstructing finite-rate-of-innovation~(FRI) signals measured through a noisy linear measurement model.
GenFRI has been recently recast as a structured low-rank optimization problem and the Cadzow projected gradient descent~(CPGD) algorithm has been suggested to solve it.
While CPGD works well in practice, only qualitative local convergence guarantees have been established.
We revisit GenFRI in the light of the regularization by denoising framework, recasting it as an optimization problem whose solutions lie in the fixed-point set of the Cadzow denoiser.
We show that no algorithm in this family can enjoy global guarantees, and establish instead that the Cadzow denoiser is quasi-nonexpansive on an explicit neighborhood of the FRI model set, whose radius is governed by the conditioning of the underlying Dirac stream.
Building on these results, we propose the generalized CPGD~(GCPGD) algorithm and prove its convergence from any initialization within an explicit basin of attraction, together with a reconstruction error bound proportional to the noise level.
Numerical experiments validate the predicted contraction rates, recovery
thresholds, and run-time certificates, and show that a single run of GCPGD
outperforms state-of-the-art GenFRI algorithms.
\end{abstract}

\IEEEpeerreviewmaketitle

\section{Introduction}
\label{section_introduction}

Pioneered in the landmark work of Vetterli, Marziliano, and Blu~\cite{vetterli2002} and further developed across various domains~\cite{dragotti2007, blu2008, uriguen2013, wei2015guaranteed}, the finite-rate-of-innovation~(FRI) framework establishes sampling and exact reconstruction principles for non-bandlimited signals defined by their so-called "rate of innovation".
Throughout this paper, the signal of interest is a $1$-periodic stream of $K$ Diracs, $x(t) = \sum_{k=1}^{K} a_k \delta(t - t_k)$, with unknown non-zero amplitudes $a_k \in \C \setminus \{0\}$ and locations $t_k \in [0, 1)$: a prototypical FRI signal with $2K$ degrees of freedom per period.
Such a stream is uniquely characterized by the vector $\vec{x} \in \C^N$ of its $N = 2M+1$ lowest Fourier series coefficients $\hat{x}_m = \sum_{k=1}^{K} a_k e^{-j2\pi m t_k}$, $-M \leq m \leq M$, provided $M \geq K$, from which the Dirac locations can be retrieved via spectral estimation~\cite{vetterli2002, blu2008}.

FRI signals are ubiquitous across engineering, finding applications in areas such as medical imaging~\cite{hao2005, tur2011, ongie2015}, radar~\cite{rudresh2017}, and radio astronomy~\cite{pan2017}. 
However, while classical FRI techniques assume direct or uniformly filtered observations of the Fourier series coefficients $\vec{x}$, many real-world acquisition systems introduce complex linear distortions, irregular sampling grids, or hardware-specific transfer functions. 
To address these challenges, Pan~\etal~\cite{pan2017} introduced the Generalized Finite-Rate-of-Innovation~(GenFRI) framework. 
In GenFRI, the Fourier coefficient vector $\vec{x}$ is observed indirectly through a known linear measurement operator $\mat{G} \in \C^{L \times N}$, yielding the noisy measurement vector $\vec{y} = \mat{G}\vec{x} + \vec{\epsilon} \in \C^L$. In the sequel, \emph{the GenFRI problem} refers to recovering $\vec{x}$ (and thus the Dirac parameters) from $\vec{y}$, and \emph{the GenFRI algorithm} to the specific method proposed in~\cite{pan2017}.

This generalized formulation spans diverse inverse problems, including point-source recovery from non-Cartesian $u$-$v$ visibilities in radio astronomy~\cite{pan2017}, acoustic reflection estimation through spatio-temporal PSFs in ultrasound~\cite{tur2011, chernyakova2014, bezzam2018}, non-Cartesian $k$-space sub-sampling in MRI~\cite{scholefield2014}, and multipath delay-Doppler estimation in radar and wireless communications~\cite{barilan2014, gedalyahu2010}.

Because the linear operator $\mat{G}$ often degrades the underlying algebraic structure (e.g., destroying the immediate rank-deficient Toeplitz or Hankel form of the Fourier coefficients), classical subspace methods like ESPRIT or MUSIC cannot be applied directly to $\vec{y}$. 
Instead, GenFRI poses the recovery as a structured low-rank matrix approximation problem, balancing data fidelity with the non-convex rank-$K$ constraint on the Toeplitz lift $\toeplitz{P}(\vec{x})$. 
However, the original GenFRI algorithms proposed in~\cite{pan2017} rely on non-convex optimization heuristics that require multiple random initializations and sensitive hyper-parameter tuning, lacking formal convergence guarantees.

Recently, Simeoni~\etal~\cite{simeoni2021} have introduced Cadzow-projected-gradient-descent~(CPGD) a significantly simpler algorithm than the one described by Pan~\etal which performs better on GenFRI problems and comes with a local guarantee of fixed point convergence.
CPGD is a projected gradient descent algorithm in which the projection step is approximated by the Cadzow denoising algorithm, a well-known denoising technique in FRI~\cite{blu2008}.
While CPGD works well in practice and seems to converge in many cases, the theoretical convergence guarantee provided by Simeoni~\etal is not fully satisfying as it is qualitative and very difficult to verify in practice.

In a complementary direction, learning-based FRI reconstruction methods have recently been proposed to overcome the breakdown of subspace-based estimation at low signal-to-noise ratio~\cite{leung2023}; the present analysis provides, within the model-based family, an explicit expression for this breakdown threshold in terms of the separation and dynamic range of the Diracs.
The recent DeepFRI technique~\cite{kamath2025deepfri} instantiates the same plug-and-play template with a learned denoiser and deploys the Cadzow denoiser within an ensemble for the high-signal-to-noise regime; since its convergence analysis relies on a Lipschitz denoiser -- a property the Cadzow denoiser lacks globally (Section~\ref{section_theory}) -- the localized guarantees developed here are complementary, certifying precisely the Cadzow phase of such hybrid schemes.

In this work, we revisit the GenFRI problem through the lens of the Regularization by Denoising~(RED) framework~\cite{reehorst2019, cohen2021}, recasting reconstruction as an optimization problem over the fixed-point set of the Cadzow denoiser. By addressing the structural limitations that prevent global convergence in this non-convex setting, we introduce the Generalized Cadzow Projected Gradient Descent~(GCPGD) algorithm and establish the first quantitative local convergence theory for GenFRI. 
The main contributions of this paper are fourfold:
\begin{itemize}
	\item We prove that global regularity hypotheses required by classical fixed-point and RED-based frameworks are structurally unattainable for Cadzow denoising, demonstrating that its rank projection is discontinuous and its fixed-point set is inherently non-convex.
	\item We establish that the Cadzow denoiser is locally quasi-nonexpansive within an explicit spectral tube surrounding the FRI model set, and we introduce a lightweight \emph{a posteriori} run-time certificate based on singular value gap ratios to replace existential localization radii.
	\item We propose the GCPGD algorithm equipped with $\mat{\Gamma}$-metric preconditioning, proving that it converges linearly from any initialization within an explicit basin of attraction to an asymptotic reconstruction error floor proportional to the noise level.
	\item We combine our operator-theoretic basin bounds with sharp Vandermonde conditioning inequalities to derive explicit, computable signal-to-noise ratio thresholds expressed directly in terms of physical Dirac parameters, including minimum separation and amplitude dynamic range.
\end{itemize}

The remainder of this paper is organized as follows. 
Section~\ref{section_theory} demonstrates why Cadzow denoising lacks global regularity and proves its local quasi-nonexpansiveness on an explicit neighborhood of the FRI model set. 
Section~\ref{section_theory_cpgd} introduces the preconditioned GCPGD algorithm, derives its basin-of-attraction convergence guarantees, and presents explicit physical SNR thresholds. 
Section~\ref{section_exp} provides extensive numerical simulations validating the theoretical bounds, inner contraction rates, and phase diagrams. 
Finally, Section~\ref{section_conclusion} concludes the paper and outlines directions for future research.

\section{On Quasi-Nonexpansiveness of Cadzow Denoising}
\label{section_theory}
\subsection{Preliminaries}
\label{subsection_preliminaries}
We consider a Hilbert space of complex vectors $\C^N$, for some integer $N = 2M +1, \; M \in \N_+$ and a Hilbert space of complex matrices $ \C^{N-P \times P+1}$, for some integer $P \leq M$.
For a matrix $\mat{H} \in \C^{N-P \times P+1}$, we define its singular value decomposition as $\mat{H} = \mat{U}\mat{\Lambda}\mat{V}^\hermtransp$, where $\mat{U} \in \C^{N-P \times N-P}$ and $\mat{V} \in \C^{P+1 \times P+1}$ are unitary matrices and $\mat{\Lambda} \in \C^{N-P \times P+1}$ is a rectangular matrix with non-negative values $\sigma_i \left( \mat{H} \right)$, $i$ between $1$ and $\min\left(P+1, N-P\right)$, on the diagonal.
We further assume that the singular values of $\mat{H}$ are sorted in a descending order such that $\sigma_i \left(\mat{H}\right) \geq \sigma_{i+1} \left(\mat{H}\right)$.

We consider the subspace $\mathbb{T}_P$ of Toeplitz matrices of $\C^{N-P \times P+1}$.
We introduce the so-called Toeplitzification operator $\toeplitz{P} : \C^N \rightarrow \mathbb{T}_P$, defined in~\cite{simeoni2021}, which constructs a Toeplitz matrix from a generator vector.
We also define the pseudoinverse of the Toeplitzification operator $T_P^\dagger:\C^{\left(N-P\right) \times \left(P+1\right)}\rightarrow \C^N$ which maps a Toeplitz matrix onto its generator as follows:
\begin{equation}
	\label{eq_pseudoinverse}
	\toeplitz{P}^\dagger=\mat{\Gamma}\toeplitz{P}^\ast,
\end{equation}
where $T_P^\ast:\C^{\left(N-P\right) \times \left(P+1\right)}\rightarrow \C^N$ is the adjoint of the Toeplizification operator~(see~\cite{simeoni2021})
and $\mat{\Gamma} \in \R^{N \times N}$ is a diagonal matrix with entries given by
\begin{align}
	\Gamma_{ii} = \min \left(i, P+1, N+1 - i \right),\quad i=1,\ldots,N.
	\label{gamma_matrix}
\end{align}

The matrix $\mat{\Gamma}$ induces a weighted inner product on $\C^N$, namely $\ginner{\vec{u}}{\vec{v}} := \vec{v}^\hermtransp \mat{\Gamma} \vec{u}$, with associated norm $\gnorm{\vec{u}} = \twonorm{\mat{\Gamma}^{1/2}\vec{u}}$.
Since each entry $x_i$ of a generator vector appears exactly $\Gamma_{ii}$ times in $\toeplitz{P}\left(\vec{x}\right)$, we have $\fronorm{\toeplitz{P}\left(\vec{x}\right)} = \gnorm{\vec{x}}$, \ie $\toeplitz{P}$ is an isometry from $\left(\C^N, \gnorm{\cdot}\right)$ onto $\left(\mathbb{T}_P, \fronorm{\cdot}\right)$, and $\toeplitz{P}\pinv$ restricted to $\mathbb{T}_P$ is its inverse.
Since $1 \leq \Gamma_{ii} \leq P+1$, the norms $\gnorm{\cdot}$ and $\twonorm{\cdot}$ are equivalent, and all convergence statements transfer between them.
Unless otherwise stated, $\C^N$ is equipped with $\gnorm{\cdot}$ in the remainder of the paper.

Using the operators introduced before, we define the orthogonal projection onto the subspace $\mathbb{T}_P$ as
\begin{equation}
	\Pi_{\mathbb{T}_P}=\toeplitz{P}\toeplitz{P}^\dagger.
	\label{proj_tp_fro}
\end{equation}

We further consider the subset $\mathcal{H}_K$ of $\C^{\left(N-P\right) \times \left(P+1\right)}$ composed of matrices with rank at most $K$.
The projection operator $\Pi_{\mathcal{H}_K}$ onto the space $\mathcal{H}_K$ of matrices with rank at most $K$ is given by the Eckart--Young--Mirsky theorem~\cite{eckart1936}.

Finally, we introduce the quantities that localize our analysis.
\begin{definition}[Gap Ratio, Model Set, Tube]
	\label{def_tube}
	For $\mat{X} \in \C^{\left(N-P\right) \times \left(P+1\right)}$ with $\sigma_K\left(\mat{X}\right) > 0$, we define the gap ratio
	\begin{equation}
		\rho\left(\mat{X}\right) := \frac{\sigma_{K+1}\left(\mat{X}\right)}{\sigma_{K}\left(\mat{X}\right)} \in \left[0, 1\right].
		\label{gap_ratio}
	\end{equation}
	We define the FRI model set $\Mod$ as the set of $\vec{x} \in \C^N$ with entries $\hat{x}_m = \sum_{k=1}^{K} a_k e^{-j 2 \pi m t_k}$, $-M \leq m \leq M$, for distinct $t_k \in \left[0, 1\right)$ and $a_k \in \C \setminus \left\lbrace 0 \right\rbrace$; the rank variety $\mathcal{X} := \left\lbrace \vec{x} \in \C^N \suchthat \rank \toeplitz{P}\left(\vec{x}\right) \leq K \right\rbrace = \toeplitz{P}^{-1}\left(\mathbb{T}_P \cap \mathcal{H}_K\right) \supsetneq \Mod$, whose nodes range over $\C \setminus \left\lbrace 0 \right\rbrace$ while those of $\Mod$ lie on the unit circle (a real codimension of $K$); and, for $\delta \in \left(0, 1\right)$, the tube
	\begin{equation}
		\mathcal{G}_\delta := \left\lbrace \mat{X} \suchthat \sigma_K\left(\mat{X}\right)>0, \; \rho\left(\mat{X}\right) \leq 1 - \delta \right\rbrace.
		\label{tube}
	\end{equation}
\end{definition}

Leaving the tube is precisely the \emph{subspace-swap} event identified by Wei and Dragotti~\cite{wei2015guaranteed} as the breakdown mechanism of subspace-based FRI estimation: $\mathcal{G}_\delta$ is a deterministic formalization of their no-swap region.

\subsection{On Quasi-Nonexpansiveness of Cadzow Denoising}
\label{subsection_cadzow}
The Cadzow algorithm is a well-known denoising procedure for FRI signals introduced by Blu~\etal~\cite{blu2008}, cast into the primal-dual optimization framework by Condat and Hirabayashi~\cite{condat2015} and recently discussed in terms of alternating projections by Simeoni~\etal~\cite{simeoni2021}.
Due to remarkable annihilation properties, the Toeplitz matrix $\toeplitz{P} \left(\vec{x}\right)$, built from the Fourier series coefficients of a stream of K Diracs, has rank $K$.
$\toeplitz{P} \left(\vec{x}\right)$ therefore belongs to a very specific subspace of $\C^{N-P \times P+1}$ \ie the intersection between $\mathbb{T}_P$ and $\mathcal{H}_K$; equivalently, $\vec{x}$ belongs to the model set $\Mod$ of Definition~\ref{def_tube}.
It is relatively straightforward to understand that in noisy scenarios, the FRI signal loses its structural property such that the matrix $\toeplitz{P} \left(\vec{x}\right)$ no longer has rank $K$.
We therefore define the following denoising problem associated with FRI:
\begin{equation}
	\min_{\mat{M} \in \mathbb{T}_{P}, \rank \left(M\right) \leq K} \fronorm{\mat{M} - \toeplitz{P}\left(\vec{x}\right)},
	\label{pb_cadzow}
\end{equation}
which has been extensively studied in the context of structured low rank approximation~\cite{chu2003}.
Problem~\ref{pb_cadzow} exhibits nice properties in the context of FRI as it is always solvable (see Theorem~\num{3.1} in~\cite{chu2003}).
Indeed its feasible set is non-empty as long as $\min \left(N - P, P+1\right) \geq K$ (see Theorem~\num{2.1} in~\cite{brualdi1993}), which could always be achieved in practice.
Its feasible set is $\mathbb{T}_P \cap \mathcal{H}_K = \toeplitz{P}\left(\mathcal{X}\right)$, not $\toeplitz{P}\left(\Mod\right)$: Cadzow denoising relaxes the unit-modulus node constraint of the model set, which the terminal parameter extraction $t_k = \arg z_k$ implicitly restores.

Equipped with the operators introduced in Section~\ref{subsection_preliminaries}, we define the Cadzow denoiser $H_n : \C^{N} \rightarrow \C^{N}$ as in~\cite{simeoni2021}:
\begin{equation}
	H_n \left( \vec{x}\right) = \toeplitz{P}\pinv \left( \Pi_{\mathbb{T}_P} \Pi_{\mathcal{H}_K}\right)^n \toeplitz{P} \left( \vec{x}\right), \; n>0.
	\label{cadzow}
\end{equation}

As highlighted by Equation~\eqref{cadzow}, Cadzow denoising is an alternating projection method between $\mathbb{T}_P$ and $\mathcal{H}_K$, which aims at finding a solution to Problem~\ref{pb_cadzow}.
We further define the set of fixed points of the Cadzow denoiser as
$\Fix \left(H_n\right) = \left\lbrace \vec{x} \in \C^N \suchthat H_n \left( \vec{x}\right) = \vec{x} \right\rbrace$.
Since every point of the rank variety is fixed by both projections and $\toeplitz{P}\pinv\toeplitz{P}$ is the identity, we have $\Mod \subsetneq \mathcal{X} \subseteq \Fix\left(H_n\right) \neq \emptyset$.
Cadzow denoising has been extensively studied in the context of alternating projections methods on manifolds~\cite{lewis2008alternating, andersson2013} and local convergence guarantees have been established.
The intricate part is the projection onto $\mathcal{H}_K$ since it is a closed but non-convex subset of $\C^{N-P \times P+1}$.
Hence, the projection onto $\mathcal{H}_K$ may not be unique for some matrices and theoretical results are difficult to derive.

The proposition hereafter provides a relatively strong result related to the projection onto $\mathcal{H}_K$.

\begin{proposition}
	\label{prop_1}
	Consider positive integers $m, n$ and $K \leq \min\left(m, n\right)$.
	For any $\mat{X} \in \C^{m \times n}$ with $\sigma_{K+1}\left(\mat{X}\right) < \sigma_{K}\left(\mat{X}\right)$, the projection of $\mat{X}$ onto $\mathcal{H}_K$ exists, is unique and is given by \begin{equation}
		\Pi_{\mathcal{H}_K} \mat{X} = \sum\limits_{i=1}^{\min \left(K, \rank \mat{X} \right)} \sigma_i \left(\mat{X}\right) \vec{u}_i \vec{v}_i^H.
		\label{KSVD}
	\end{equation}
\end{proposition}
\noindent\emph{The proof is provided in Section~\ref{sup_section_prop_1} in the supplementary material.}

Contrary to previous analyses, we do not attempt to establish any global regularity of the Cadzow denoiser.
The following remark shows that no such property can hold.

\begin{remark}[Discontinuity of the Rank Projection]
	\label{remark_discontinuity}
	$\Pi_{\mathcal{H}_K}$ is discontinuous across the set $\left\lbrace \sigma_K = \sigma_{K+1}\right\rbrace$: for $K=1$, the matrices $\mat{X}_\varepsilon = \textnormal{diag}\left(1+\varepsilon, 1\right)$ and $\mat{Z}_\varepsilon = \textnormal{diag}\left(1, 1+\varepsilon\right)$ satisfy $\fronorm{\mat{X}_\varepsilon - \mat{Z}_\varepsilon} = \sqrt{2}\varepsilon$ while $\fronorm{\Pi_{\mathcal{H}_1}\mat{X}_\varepsilon - \Pi_{\mathcal{H}_1}\mat{Z}_\varepsilon} = \sqrt{2}\left(1+\varepsilon\right)$, so that no global Lipschitz constant exists for $\Pi_{\mathcal{H}_K}$, hence none for $H_n$.
	No tie-breaking selection removes the jump, since $\Pi_{\mathcal{H}_K}$ is uniquely determined on either side of the set $\left\lbrace \sigma_K = \sigma_{K+1}\right\rbrace$: global regularity of the Cadzow denoiser is unavailable, and the analysis must be confined to the tube $\mathcal{G}_\delta$, where the spectral gap $\sigma_K - \sigma_{K+1}$ is bounded away from zero.
\end{remark}

We consequently quantify the behavior of the alternating projections on the tube $\mathcal{G}_\delta$.
The following lemma, which relies on the prox-regularity of the rank variety established by Luke~\cite{luke2013}, provides an unconditional one-step estimate.

\begin{lemma}[One-step Estimate]
	\label{lemma_one_step}
	Let $\delta \in \left(0, 1\right)$, $\mat{X} \in \mathcal{G}_\delta$ and $\mat{M}^\star \in \mathbb{T}_P \cap \mathcal{H}_K$.
	Then $\hat{\mat{X}} := \Pi_{\mathcal{H}_K}\left(\mat{X}\right)$ is uniquely defined and
	\begin{multline}
		\left(1 - \rho\left(\mat{X}\right)\right)\fronorm{\Pi_{\mathbb{T}_P}\hat{\mat{X}} - \mat{M}^\star}^2 \leq \fronorm{\mat{X} - \mat{M}^\star}^2 \\ - \fronorm{\mat{X} - \hat{\mat{X}}}^2 - \left(1 - \rho\left(\mat{X}\right)\right)\fronorm{\hat{\mat{X}} - \Pi_{\mathbb{T}_P}\hat{\mat{X}}}^2.
		\label{one_step}
	\end{multline}
\end{lemma}
\noindent\emph{The proof is provided in Section~\ref{sup_section_lemma_1} in the supplementary material.}

Lemma~\ref{lemma_one_step} establishes quasi-nonexpansiveness of the alternating projection step up to the factor $\left(1 - \rho\left(\mat{X}\right)\right)^{-1} \in \left[1, \delta^{-1}\right]$, which is exact only as $\rho \rightarrow 0$.

To upgrade it to a quantitative local statement, we now describe the geometry of the intersection $\mathbb{T}_P \cap \mathcal{H}_K$ around a model point.

Let $\mat{M}^\star = \mat{U}\mat{\Sigma}\mat{V}^\hermtransp \in \mathbb{T}_P \cap \mathcal{H}_K$ have rank exactly $K$, with $\mat{U}$ and $\mat{V}$ collecting its $K$ leading singular vectors.
By~\cite{luke2013}, $\mathcal{H}_K$ is prox-regular at $\mat{M}^\star$ and its normal cone is the linear space
\begin{equation}
	N_{\mathcal{H}_K}\left(\mat{M}^\star\right) = \left\lbrace \mat{W} \suchthat \mat{U}^\hermtransp\mat{W} = \mat{0}, \; \mat{W}\mat{V} = \mat{0} \right\rbrace.
	\label{normal_cone}
\end{equation}

Near $\mat{M}^\star$, the set $\mathcal{H}_K$ is the smooth manifold of rank-$K$ matrices, with tangent space $T_{\mat{M}^\star}\mathcal{H}_K = N_{\mathcal{H}_K}\left(\mat{M}^\star\right)^\perp$.
By a Kronecker-type theorem, the Toeplitz matrices of rank at most $K$ are, up to degenerate strata (confluent or boundary nodes), generated by $K$-term exponential models $\hat{x}_m = \sum_{k=1}^{K} c_k z_k^m$, whose bases $z_k \in \C \setminus \left\lbrace 0 \right\rbrace$ are called the \emph{nodes}; on the model set $\Mod$ the nodes are unit-modulus, $z_k = e^{-j2\pi t_k}$, in bijection with the Dirac locations.
Around points with $K$ pairwise distinct nodes and nonzero coefficients -- a dense open stratum containing $\toeplitz{P}\left(\Mod\right)$ -- the intersection $\mathbb{T}_P \cap \mathcal{H}_K$ is a smooth complex manifold $\mathcal{E}$ of dimension $2K$, whose tangent space $T_{\mathcal{E}}\left(\mat{M}^\star\right)$ is spanned by the derivatives of the exponential model with respect to its nodes and coefficients~\cite{andersson2013}. 
A transversal intersection of $\mathbb{T}_P$ and $\mathcal{H}_K$ would have complex dimension $N - \left(N-P-K\right)\left(P+1-K\right)$, which is smaller than $2K$ as soon as $K \leq \min\left(N-P, P+1\right) - 2$: in the denoising regime $K < P$ of interest, the intersection is excess-dimensional, the two sets cannot meet transversally at any model point, and the constraint qualifications of~\cite{lewis2008alternating} are unavailable.
The adequate regularity notion is the \emph{nontangentiality} of Andersson and Carlsson~\cite{andersson2013}.

\begin{definition}[Nontangentiality]
	\label{def_nontangential}
	A rank-$K$ point $\mat{M}^\star \in \mathbb{T}_P \cap \mathcal{H}_K$ with $K$ distinct nodes is \emph{nontangential} if $\mathbb{T}_P \cap T_{\mat{M}^\star}\mathcal{H}_K = T_{\mathcal{E}}\left(\mat{M}^\star\right)$; its angle constant is then $c\left(\mat{M}^\star\right) := \cos\alpha\left(\mat{M}^\star\right) \in \left[0, 1\right)$, where $\alpha\left(\mat{M}^\star\right)$ is the first principal angle between $\mathbb{T}_P \ominus T_{\mathcal{E}}\left(\mat{M}^\star\right)$ and $T_{\mat{M}^\star}\mathcal{H}_K \ominus T_{\mathcal{E}}\left(\mat{M}^\star\right)$.
\end{definition}

Nontangentiality imposes no restriction on the dimensions of the two sets~\cite{andersson2013}, and it is verifiable for any given instance. Since $T_{\mathcal{E}} \subseteq \mathbb{T}_P \cap T_{\mat{M}^\star}\mathcal{H}_K$ always holds, $\mat{M}^\star$ is nontangential if and only if exactly $2K$ principal angles between $\mathbb{T}_P$ and $N_{\mathcal{H}_K}\left(\mat{M}^\star\right)^\perp$ vanish, and $c\left(\mat{M}^\star\right)$ is the cosine of the next one, computable from a single singular value decomposition.
Every transversal point is nontangential, and for the square Hankel counterpart of the Toeplitz lift -- to which the present setting reduces by row reversal -- nontangentiality is established at all model points outside a thin exceptional set in~\cite{andersson2013}.
The angle constant $c\left(\mat{M}^\star\right)$ remains bounded away from $1$ as Diracs collide -- nontangentiality does not degrade in the collision limit -- whereas the conditioning $\sigma_K$ vanishes there; the two geometric quantities are independent, and it is $\sigma_K$, not $c$, that governs the basin.

\begin{lemma}[Local Convergence of Alternating Projections]
\label{lemma_nontangential}
Let $\mat{M}^\star \in \mathbb{T}_P \cap \mathcal{H}_K$ be a nontangential point with angle constant $c := c(\mat{M}^\star) < 1$, and fix $\bar{c} \in (c^2, 1)$ and $\varepsilon \in (0, 1]$.
There exists $r_0 > 0$ such that, for every initial matrix $\mat{M}_0$ satisfying $\|\mat{M}_0 - \mat{M}^\star\|_F \le r_0$, the alternating projections $\mat{M}_{k+1} = \Pi_{\mathbb{T}_P}\Pi_{\mathcal{H}_K}\mat{M}_k$ are well defined, converge to a point in $\mathbb{T}_P \cap \mathcal{H}_K$, and satisfy for every $k \in \mathbb{N}$:

\begin{enumerate}
	\item \emph{Iterate Stability:}
	\begin{equation}
		\|\mat{M}_k - \mat{M}^\star\|_F \le (1+\varepsilon)\|\mat{M}_0 - \mat{M}^\star\|_F.
		\label{ap_estimate_1}
	\end{equation}
	
	\item \emph{Linear Model Contraction:}
	\begin{equation}
		\operatorname{dist}\left(\mat{M}_k, \mathbb{T}_P \cap \mathcal{H}_K\right) \le \bar{c}^{\,k} \operatorname{dist}\left(\mat{M}_0, \mathbb{T}_P \cap \mathcal{H}_K\right).
		\label{ap_estimate_2}
	\end{equation}
\end{enumerate}
\end{lemma}
\noindent\emph{The proof is provided in Section~\ref{sup_section_lemma_2} in the supplementary material.}

Lemma~\ref{lemma_nontangential} controls both the drift of the inner Cadzow iterates and the contraction of their distance to the model variety. It yields the following theorem.

\begin{myth}
	\label{th_1}
	Let $\delta \in (0, 1)$, $n \ge 1$, and $\vec{x}^\star \in \mathcal{M}$ with $\toeplitz{P}(\vec{x}^\star)$ of rank $K$, nontangential with angle constant $c := c(\toeplitz{P}(\vec{x}^\star)) < 1$, and singular value $\sigma_K := \sigma_K(\toeplitz{P}(\vec{x}^\star)) > 0$. Fix $\bar{c} \in (c^2, 1)$, $\varepsilon \in (0, 1]$, and define the spectral tube radius:
	\begin{equation}
		r_\delta := \frac{1-\delta}{2-\delta}\, \sigma_K.
		\label{eq:tube_radius}
	\end{equation}
	Let $r_0 > 0$ denote the nontangential radius from Lemma~\ref{lemma_nontangential}, and set $r := \min\left(\frac{r_\delta}{1+\varepsilon}, r_0\right)$. For every signal $\vec{x} \in \mathbb{C}^N$ satisfying $\|\vec{x} - \vec{x}^\star\|_{\boldsymbol{\Gamma}} \le r$:
	\begin{enumerate}
		\item \emph{Forward Tube Invariance:} The inner alternating projection iterate $\mat{M}_k := (\Pi_{\mathbb{T}_P}\Pi_{\mathcal{H}_K})^k \toeplitz{P}(\vec{x})$ is such that
		\begin{equation}
			\mat{M}_k \in \mathcal{G}_\delta \; \text{and} \; \|\mat{M}_k - \toeplitz{P}(\vec{x}^\star)\|_F \le r_\delta, \; k \in \left\lbrace 0, \ldots, n \right\rbrace .
		\end{equation}
		
		\item \emph{Approximate Quasi-Nonexpansiveness:}
		\begin{equation}
			\|H_n(\vec{x}) - \vec{x}^\star\|_{\mat{\Gamma}} \le (1+\varepsilon)\|\vec{x} - \vec{x}^\star\|_{\mat{\Gamma}}.
			\label{eq:quasi_nonexpansive}
		\end{equation}
		
		\item \emph{Linear Model Contraction:}
		\begin{equation}
			\operatorname{dist}\left(\toeplitz{P}(H_n(\vec{x})), \mathbb{T}_P \cap \mathcal{H}_K\right) \le \bar{c}^n \operatorname{dist}\left(\toeplitz{P}(\vec{x}), \mathbb{T}_P \cap \mathcal{H}_K\right).
			\label{eq:model_contraction}
		\end{equation}
	\end{enumerate}
	In particular, $H_n$ is approximately quasi-nonexpansive at nontangential model points—becoming exact in the limit $\varepsilon \to 0$—and contracts the distance to the model variety at the per-cycle linear rate $\bar{c} \in (c^2, 1)$.
\end{myth}

\noindent\emph{The proof is provided in Section~\ref{sup_section_th_1} in the supplementary material.}

Two observations can be drawn from this Theorem.
First, the effective basin radius $r = \min(r_\delta/(1+\varepsilon), r_0)$ scales directly with Dirac stream conditioning via $r_\delta \propto \sigma_K$. As Diracs collide, the basin closes due to $\sigma_K \to 0$ rather than any loss of nontangentiality ($r_0$ and $c$ remain bounded).

\begin{remark}[A posteriori Drift Certificate]
	\label{remark_certificate}
	While the basin radius $r_0$ in Lemma~\ref{lemma_nontangential} is non-constructive, Lemma~\ref{lemma_one_step} provides an online, zero-cost certificate for the error drift. By dropping the negative terms in~\eqref{one_step}, each inner Cadzow step expands the error by at most $(1 - \rho(\mat{M}_k))^{-1/2}$. Accumulating this bound across $n$ iterations yields
	\begin{equation}
	\gnorm{H_n\left(\vec{x}\right) - \vec{x}^\star} \leq \left(1 + \hat{\varepsilon}\left(\vec{x}\right)\right)\gnorm{\vec{x} - \vec{x}^\star},
	\end{equation}
	with
	\begin{equation}
		\qquad 1 + \hat{\varepsilon}\left(\vec{x}\right) := \prod_{k=0}^{n-1}\left(1 - 	\rho\left(\mat{M}_k\right)\right)^{-1/2}.
		\label{certificate}
	\end{equation}
	Because Cadzow denoising computes the singular values of $\mat{M}_k$ at each step anyway, the gap ratios $\rho(\mat{M}_k) = \sigma_{K+1}/\sigma_K$ and the factor $\hat{\varepsilon}(\vec{x})$ are obtained at zero additional cost. Consequently, observing $\hat{\varepsilon}(\vec{x}) \le \varepsilon$ during execution validates the drift bound of Theorem~\ref{th_1} without requiring prior knowledge of $r_0$.
\end{remark}

Second, measuring distances in the weighted $\mat{\Gamma}$-metric eliminates the $\sqrt{P+1}$ denoiser scaling used in~\cite{simeoni2021}, ensuring the implemented and analyzed operators coincide directly without requiring proximal modifications.

\section{Generalized Cadzow-Projected Gradient Descent}
\label{section_theory_cpgd}
The fixed-point perspective of Section~\ref{section_theory} places GenFRI squarely within the RED-PRO framework, introduced by Cohen~\etal~\cite{cohen2021} for inverse problems in image processing which minimizes a data-fidelity term over the fixed-point set of a denoiser.
They suggest several variants of the hybrid steepest descent~(HSD) algorithm in order to solve the RED-PRO template.
They further demonstrate that providing $d$-demicontractivity of the denoiser and proper convexity of the loss, the proposed HSD algorithms converge to an optimal solution of the RED-PRO template.

GenFRI instantiates this template exactly. With the Cadzow denoiser $H_n$ and the quadratic loss, the template becomes the fixed-point formulation
\begin{equation}
	\min_{\vec{x} \in \C^N} \twonorm{\mat{G} \vec{x} - \vec{y}}^2 \textnormal{ s.t. } \vec{x} \in \Fix\left(H_n\right).
	\label{redpro}
\end{equation}
The fit is, however, only structural, and the full power of the framework cannot be unlocked. Indeed, the following remark shows that the fixed-point set of the Cadzow denoiser -- hence Problem~\ref{redpro} -- is non-convex, so that the $d$-demicontractivity underpinning the convergence theory of~\cite{cohen2021} fails globally.

\begin{remark}[Non-convexity of the Fixed-point Set]
	\label{remark_nonconvex}
	Any operator $T$ that is globally $d$-demicontractive ($d < 1$) must have a convex fixed-point set $\Fix(T)$, as the condition $\|T(\vec{x}) - \vec{y}\|^2 \le \|\vec{x} - \vec{y}\|^2$ is affine in $\vec{y}$ and defines an intersection of closed halfspaces.
	However, the FRI model set $\Mod \subseteq \Fix(H_n)$ is non-convex: averaging two generic $K$-Dirac streams yields a Toeplitz matrix of rank $2K$.
	Consequently, $H_n$ cannot be globally demicontractive—regardless of scaling or relaxation—and the fixed-point problem~\eqref{redpro} is inherently non-convex.
	Together with Remark~\ref{remark_discontinuity}, this demonstrates why global RED-PRO hypotheses~\cite{cohen2021} fail here, necessitating a localized analysis around $\Mod$.
\end{remark}

We note that the feasible set of the CPGD optimization problem described in Eq.~(17) of~\cite{simeoni2021},
\begin{equation}
	\min_{\vec{x} \in \C^N} \twonorm{\mat{G} \vec{x} - \vec{y}}^2 + \iota_{\mathcal{X}} \left(\vec{x}\right),
	\label{cpgd}
\end{equation}
where $\iota_{\mathcal{X}}$ denotes the indicator function of the rank variety $\mathcal{X}$ of Definition~\ref{def_tube}, is precisely $\mathcal{X}$, so that any point feasible for Problem~\ref{cpgd} is feasible for Problem~\ref{redpro} by the inclusion $\mathcal{X} \subseteq \Fix\left(H_n\right)$.
Rather than forcing GenFRI into the framework of~\cite{cohen2021}, we use it as motivation only -- the analysis below is self-contained and never invokes the convergence theory of~\cite{cohen2021} -- and consider the localized GenFRI problem
\begin{equation}
	\min_{\vec{x} \in \Fix\left(H_n\right) \cap \mathcal{B}_\delta\left(\vec{x}^\star\right)} \twonorm{\mat{G} \vec{x} - \vec{y}}^2,
	\label{redproloc}
\end{equation}
where $\mathcal{B}_\delta\left(\vec{x}^\star\right) := \left\lbrace \vec{x} \in \C^N \suchthat \gnorm{\vec{x} - \vec{x}^\star} \leq r \right\rbrace$. 
Note that Problem~\ref{redproloc} is anchored at the unknown $\vec{x}^\star$ and is therefore an analysis device rather than an operational program: part~(3) of Theorem~\ref{th_2} below is a stability statement about its solution set, complementary to (and logically independent of) the confinement of parts~(1)--(2), and Remark~\ref{remark_feasibility} characterizes the fixed points the algorithm actually reaches -- with $r$ the radius of Theorem~\ref{th_1}, and $\vec{x}^\star \in \Mod$ denotes the underlying FRI signal.

We now introduce GCPGD, a slightly more general algorithm than CPGD, described in Algorithm~\ref{algo_CPGD}.
\begin{algorithm}[htb]
	\caption{Generalized CPGD~(GCPGD)}
	\label{algo_CPGD}
	\begin{algorithmic}
		\REQUIRE $\vec{y},\, \mat{G},\, \toeplitz{P},\, \vec{x}_{0},\,K\leq P,\left\lbrace\tau_k\right\rbrace_{k \in \N}, n\in\N, \alpha \in \left(0, 1\right]$\\
		\STATE k:=0
		\REPEAT
		\STATE ${\vec{v}}_{k+1} := \vec{x}_{k} - 2\tau_k \mat{\Gamma}^{-1}\mat{G}^H \left(\mat{G}\vec{x}_{k}-\vec{y}\right) $
		\STATE ${\vec{z}}_{k+1}:=H_n\left({\vec{v}}_{k+1}\right)$
		\STATE ${\vec{x}}_{k+1}:=\alpha {\vec{z}}_{k+1} + \left(1 - \alpha \right) {\vec{v}}_{k+1}$
		\STATE $k \leftarrow k + 1$
		\UNTIL{a stopping criterion is satisfied}
		\RETURN ${\vec{x}}_k$
	\end{algorithmic}
\end{algorithm}

The gradient step of Algorithm~\ref{algo_CPGD} is taken in the weighted metric: $\mat{\Gamma}^{-1}\mat{G}^H\left(\mat{G}\cdot - \vec{y}\right)$ is the $\mat{\Gamma}$-gradient of the least-squares objective $\vec{x} \mapsto \twonorm{\mat{G}\vec{x} - \vec{y}}^2$, \ie the Riesz representation of its differential with respect to $\ginner{\cdot}{\cdot}$.
This preconditioning makes the associated Hessian $\mat{\Gamma}^{-1}\mat{G}^H\mat{G}$ self-adjoint and positive semidefinite for $\ginner{\cdot}{\cdot}$, so that the gradient step is $\mat{\Gamma}$-nonexpansive and no metric-mismatch factor arises; the unpreconditioned step $\mat{G}^H\left(\mat{G}\cdot - \vec{y}\right)$ of CPGD~\cite{simeoni2021} is recovered when $\mat{\Gamma} = Id$.

In the sequel, $\gnormop{\mat{G}}$ denotes the operator norm of $\mat{G}$
from $\left(\C^N, \gnorm{\cdot}\right)$ to $\left(\C^L, \twonorm{\cdot}\right)$,
i.e.\ the largest singular value of $\mat{G}\mat{\Gamma}^{-1/2}$, and
$\mu := \inf\left\lbrace \twonorm{\mat{G}\vec{u}}/\gnorm{\vec{u}} \suchthat
\vec{u} \in \mathcal{B}_\delta\left(\vec{x}^\star\right) -
\mathcal{B}_\delta\left(\vec{x}^\star\right), \vec{u} \neq 0 \right\rbrace$
the restricted injectivity constant of $\mat{G}$ on the ball.
Since that ball is full-dimensional, the infimum is attained over all
directions and $\mu$ coincides with the global constant
$\mu_{\mat{\Gamma}} := \sigma_{\min}\left(\mat{G}\mat{\Gamma}^{-1/2}\right)$,
with which Theorem~\ref{th_2} is stated.
The algorithm explores far more structured directions -- the denoiser
confines the iterates to a neighborhood of the model set -- and the effective
constant is the tangent-restricted $\mu_T := \inf\left\lbrace
\twonorm{\mat{G}\vec{u}}/\gnorm{\vec{u}} \suchthat \vec{u} \in
T_{\vec{x}^\star}\Mod \setminus \left\lbrace 0 \right\rbrace \right\rbrace
\geq \mu$, computable from a single singular value decomposition of $\mat{G}$
on a $\mat{\Gamma}$-orthonormal basis of the model tangent and measured
throughout Section~\ref{section_exp}.

The following lemma shows that, under a joint smallness condition on the initialization error and the noise, the whole trajectory of GCPGD is confined to the ball on which Theorem~\ref{th_1} applies, and contracts geometrically to a noise-dependent floor.

\begin{lemma}[Local Confinement and Linear Convergence]
	\label{lemma_confinement}
	Let $\vec{x}^\star \in \Mod$ with singular value $\sigma_K > 0$, let measurements be $\vec{y} = \mat{G}\vec{x}^\star + \vec{\epsilon}$, and suppose the conclusions of Theorem~\ref{th_1} hold with parameters $(\varepsilon, r)$. 
	
	Fix relaxation $\alpha \in (0, 1]$, step size $\tau \in \left(0, \left(2\gnormop{\mat{G}}^2\right)^{-1}\right]$, and define the parameters:
	\begin{equation}
		q := \left(1 - 2\tau\mu^2\right)^{1/2}, \; \tilde{q} := (1+\alpha\varepsilon) q, \; \text{and} \; b := 2\tau \gnormop{\mat{G}} \twonorm{\vec{\epsilon}}.
		\label{q_and_b}
	\end{equation}
	If the contraction rate satisfies $\tilde{q} < 1$, and the initialization and noise satisfy:
	\begin{equation}
		\gnorm{\vec{x}_0 - \vec{x}^\star} \leq r \quad \text{and} \quad \frac{(1+\alpha\varepsilon)b}{1-\tilde{q}} \leq r,
		\label{confinement_cond}
	\end{equation}
	then the GCPGD sequence $(\vec{x}_k)_{k \in \N}$ and gradient steps $(\vec{v}_{k+1})_{k \in \N}$ satisfy for all $k \in \N$:
	
	\begin{enumerate}
		\item \emph{Forward Ball Invariance:} 
		\begin{equation}
			\vec{x}_k \in \mathcal{B}_\delta\left(\vec{x}^\star\right) \quad \text{and} \quad \vec{v}_{k+1} \in \mathcal{B}_\delta\left(\vec{x}^\star\right).
		\end{equation}
		
		\item \emph{Geometric Confinement \& Convergence:}
		\begin{equation}
			\gnorm{\vec{x}_k - \vec{x}^\star} \leq \tilde{q}^k \gnorm{\vec{x}_0 - \vec{x}^\star} + \frac{(1+\alpha\varepsilon)b}{1-\tilde{q}} \leq r.
			\label{confinement_bound}
		\end{equation}
	\end{enumerate}
\end{lemma}
\noindent\emph{The proof is provided in Section~\ref{sup_section_lemma_3} in the supplementary material.}

We draw the following result regarding the convergence of GCPGD.

\begin{myth}[Local Convergence and Stability of GCPGD]
	\label{th_2}
	Let $\vec{x}^\star \in \Mod$ with singular value $\sigma_K := \sigma_K\left(\toeplitz{P}\left(\vec{x}^\star\right)\right) > 0$, and let noisy measurements be given by $\vec{y} = \mat{G}\vec{x}^\star + \vec{\epsilon}$. 
	Fix parameters $\delta \in (0, 1)$ and $\alpha \in (0, 1]$, and assume that:
	\begin{enumerate}
		\item $\toeplitz{P}\left(\vec{x}^\star\right)$ is nontangential in the sense of Definition~\ref{def_nontangential};
		\item The restricted injectivity constant satisfies $\mu > 0$.
	\end{enumerate}
	There exist constants $c_1, c_2 > 0$, depending explicitly on $(\delta, \alpha, \gnormop{\mat{G}}, \mu, \tau)$ and the radius $r_0$ of Lemma~\ref{lemma_nontangential}, such that if the step size, noise level, and initialization satisfy:
	\begin{equation}
		\tau_k \equiv \tau \in \left(0, \frac{1}{2\gnormop{\mat{G}}^2}\right], \; \twonorm{\vec{\epsilon}} \leq c_1 \sigma_K, \; \gnorm{\vec{x}_0 - \vec{x}^\star} \leq c_2 \sigma_K,
	\end{equation}
	then, defining contraction rates 
	\begin{equation}
		q := \left(1-2\tau\mu^2\right)^{1/2} < 1 \; \text{and} \; \tilde{q} := \frac{1+q}{2} \in (0, 1),
	\end{equation}
	 the iterates $(\vec{x}_k)_{k \in \N}$ of Algorithm~\ref{algo_CPGD} satisfy:
	
	\begin{enumerate}
		\item \emph{Forward Invariance:} $\vec{x}_k \in \mathcal{B}_\delta\left(\vec{x}^\star\right)$ and $\vec{v}_{k+1} \in \mathcal{B}_\delta\left(\vec{x}^\star\right)$ for all $k \in \N$, and the inner Cadzow iterates applied to $\vec{v}_{k+1}$ remain inside $\mathcal{G}_\delta$.
		
		\item \emph{Geometric Confinement:} For all $k \in \N$,
		\begin{equation}
			\gnorm{\vec{x}_k - \vec{x}^\star} \leq \tilde{q}^k \gnorm{\vec{x}_0 - \vec{x}^\star} + \frac{8\gnormop{\mat{G}}}{\mu^2} \twonorm{\vec{\epsilon}},
		\end{equation}
		which yields the asymptotic noise floor:
		\begin{equation}
			\limsup_{k \to \infty} \gnorm{\vec{x}_k - \vec{x}^\star} \leq \frac{8\gnormop{\mat{G}}}{\mu^2} \twonorm{\vec{\epsilon}}.
		\end{equation}
		
		\item \emph{Local Minimizer Stability:} Every fixed point $\bar{\vec{x}} \in \Fix\left(H_n\right) \cap \mathcal{B}_\delta\left(\vec{x}^\star\right)$ solving the localized Problem~\ref{redproloc} satisfies:
		\begin{equation}
			\gnorm{\bar{\vec{x}} - \vec{x}^\star} \leq \frac{2}{\mu} \twonorm{\vec{\epsilon}}.
		\end{equation}
	\end{enumerate}
\end{myth}
\noindent\emph{The proof is provided in Section~\ref{sup_section_th_2} in the supplementary material.}

Under nontangentiality, the linear convergence of GCPGD is driven by the restricted injectivity of the measurement operator rather than by the denoiser -- the denoiser only confines the iterates to the rank variety $\mathcal{X}$, and the data term selects the model point within it: in the noiseless case $\vec{\epsilon} = \vec{0}$, part~(2) gives $\gnorm{\vec{x}_k - \vec{x}^\star} \leq \tilde{q}^k\gnorm{\vec{x}_0 - \vec{x}^\star}$ with $\tilde{q} = \frac{1+q}{2} < 1$, so that $\vec{x}_k \rightarrow \vec{x}^\star$ at a linear rate, while part~(\num{3}) of Theorem~\ref{th_1} simultaneously contracts the distance of the iterates to the model variety at the rate $\bar{c}^{\,n}$ per outer iteration.
In general the iterates converge linearly to the noise floor of part~(2), and the stability bound of part~(3) applies to any solution of Problem~\ref{redproloc}.
This linear behavior is consistent with the experiments of Section~\ref{section_exp}.

\begin{remark}[Model Consistency of Fixed Points; Roles of $n$ and $\alpha$]
	\label{remark_feasibility}
	The rate $\tilde{q}$ and the noise floor of Theorem~\ref{th_2} do not involve the inner iteration count $n$; the role of $n$ (and one role of $\alpha$) is the model consistency of the limit.
	Let $\bar{\vec{x}} \in \mathcal{B}_\delta\left(\vec{x}^\star\right)$ be a fixed point of the GCPGD iteration, $\bar{\vec{v}}$ its gradient step, and $\mathcal{X} \supsetneq \Mod$, the rank variety of Definition~\ref{def_tube}.
	Fixed-point stationarity gives $\alpha\left(H_n\left(\bar{\vec{v}}\right) - \bar{\vec{v}}\right) = 2\tau\mat{\Gamma}^{-1}\mat{G}^\hermtransp\left(\mat{G}\bar{\vec{x}} - \vec{y}\right)$, so that $\gnorm{H_n\left(\bar{\vec{v}}\right) - \bar{\vec{v}}} \leq \frac{2\tau\gnormop{\mat{G}}}{\alpha}\left(\gnormop{\mat{G}}\gnorm{\bar{\vec{x}} - \vec{x}^\star} + \twonorm{\vec{\epsilon}}\right)$, while part~(\num{3}) of Theorem~\ref{th_1} yields
	\begin{equation}
		\textnormal{dist}_{\mat{\Gamma}}\left(\bar{\vec{v}}, \mathcal{X}\right) \leq \frac{\gnorm{H_n\left(\bar{\vec{v}}\right) - \bar{\vec{v}}}}{1 - \bar{c}^{\,n}}.
		\label{feasibility_bound}
	\end{equation}
	Larger $n$ and larger $\alpha$ therefore tighten the model consistency of the limit.
\end{remark}

Theorem~\ref{th_2} provides, to the best of our knowledge, the first quantitative basin-of-attraction and noise-robustness guarantee in the GenFRI setting: the basin radius and the admissible noise level both scale with $\sigma_K$, which measures the conditioning of the Dirac stream, and the reconstruction error of the limit point is linear in the noise level.
The following corollary shows that, with the natural least-squares initialization, all hypotheses collapse into a single signal-to-noise condition.

\begin{corollary}
	\label{cor_warm_start}
	Under the setting of Theorem~\ref{th_2}, assume additionally that $\mat{G}$ has full column rank with smallest singular value $\mu_{\mat{\Gamma}} > 0$ in the weighted metric, and initialize Algorithm~\ref{algo_CPGD} at $\vec{x}_0 = \mat{G}\pinv\vec{y}$.
	Then $\gnorm{\vec{x}_0 - \vec{x}^\star} = \gnorm{\mat{G}\pinv\vec{\epsilon}} \leq \twonorm{\vec{\epsilon}} / \mu_{\mat{\Gamma}}$, so that the initialization hypothesis of Theorem~\ref{th_2} is implied by the noise hypothesis whenever $\twonorm{\vec{\epsilon}} \leq c_2\, \mu_{\mat{\Gamma}}\, \sigma_K$, and GCPGD converges after a single run under this sole signal-to-noise condition.
\end{corollary}

The hypotheses of Theorem~\ref{th_2} and Corollary~\ref{cor_warm_start} are expressed through $\sigma_K$, which is not directly a physical parameter of the GenFRI problem.
The following lemma closes this gap: it lower-bounds $\sigma_K$ in terms of the amplitudes and the separation of the Diracs, by combining the Vandermonde factorization of the Toeplitz lift with the sharp conditioning bounds of Moitra~\cite{moitra2015}.

\begin{lemma}[Conditioning of the Dirac Stream]
	\label{lemma_conditioning}
	Let $\vec{x}^\star \in \Mod$ be a signal generated by $K$ Diracs with parameters:
	\begin{itemize}
		\item Locations $t_1, \ldots, t_K \in [0, 1)$ with minimum wrap-around separation: $\Delta := \min_{k \neq l} \min\left( |t_k - t_l|, \, 1 - |t_k - t_l| \right);$
		\item Non-zero amplitudes $a_1, \ldots, a_K \in \C \setminus \{0\}$.
	\end{itemize}
	If the matrix pencil dimensions satisfy $\min(N-P, \, P+1) > 1 + \frac{1}{\Delta}$, then the $K$-th singular value $\sigma_K := \sigma_K\left(\toeplitz{P}\left(\vec{x}^\star\right)\right)$ satisfies:
	\begin{equation}
		\sigma_K \ge \left( \min_{1 \le k \le K} |a_k| \right) \sqrt{(N-P) - \frac{1}{\Delta} - 1} \; \sqrt{(P+1) - \frac{1}{\Delta} - 1}.
		\label{sigmaK_bound}
	\end{equation}
\end{lemma}
\noindent\emph{The proof is provided in Section~\ref{sup_section_lemma_4} in the supplementary material.}

Lemma~\ref{lemma_conditioning} allows us to derive the following explicit convergence conditions for GCPGD.
\begin{corollary}[Explicit Convergence Conditions]
	\label{cor_explicit}
	Under the setting of Lemma~\ref{lemma_conditioning}, assume that:
	\begin{enumerate}
		\item The dimensions satisfy $\min(N-P, \, P+1) > 1 + \frac{1}{\Delta}$;
		\item $\toeplitz{P}\left(\vec{x}^\star\right)$ is nontangential, with $r_0$ satisfying $r_0 \ge \frac{1}{2} r_\delta$;
		\item $\mat{G}$ has full column rank with minimum weighted singular value $\mu_{\mat{\Gamma}} > 0$.
	\end{enumerate}
	If Algorithm~\ref{algo_CPGD} is initialized at $\vec{x}_0 = \mat{G}\pinv\vec{y}$, then GCPGD converges linearly to a limit point $\bar{\vec{x}}$ satisfying the error bound
	\begin{equation}
		\gnorm{\bar{\vec{x}} - \vec{x}^\star} \leq \frac{2}{\mu_{\mat{\Gamma}}} \twonorm{\vec{\epsilon}},
	\end{equation}
	whenever the noise level obeys the explicit signal-to-noise threshold
	\begin{equation}
		\label{explicit_snr}
		\begin{aligned}
			\twonorm{\vec{\epsilon}} \leq{} & c_2 \, \mu_{\mat{\Gamma}} \left( \min_{1 \le k \le K} |a_k| \right) \sqrt{(N-P) - \frac{1}{\Delta} - 1} \\
			& \times \sqrt{(P+1) - \frac{1}{\Delta} - 1},
		\end{aligned}
	\end{equation}
	with constant $c_2 := \frac{1-\delta}{2(2-\delta)}$.
\end{corollary}
\noindent\emph{The proof is provided in Section~\ref{sup_section_corr_1} in the supplementary material.}

Corollary~\ref{cor_explicit} expresses the convergence of GCPGD directly in terms of the physical parameters of the GenFRI problem.
The separation condition $\Delta > 1/\left(\min\left(N-P,P+1\right) - 1\right)$ is of Rayleigh type -- the Diracs must be resolvable at the resolution of the smaller dimension of the Toeplitz lift, which is maximized by the  choice $P = M$ used in Section~\ref{section_exp} -- and the admissible noise level grows with the oversampling and shrinks with the amplitude dynamic range through $\min_k\lvert a_k \rvert$.
With the noise model of Section~\ref{section_exp}, where the $\epsilon_l$ have standard deviation $\sigma_n = \max_k \lvert a_k\rvert\, e^{-\textnormal{PSNR}/10}$ so that $\twonorm{\vec{\epsilon}} \approx \sigma_n \sqrt{L}$, condition~\eqref{explicit_snr} reads, up to noise concentration,
\begin{equation}
	\textnormal{PSNR} \;\geq\; 10 \ln \frac{R\sqrt{L}}{c_2\, \mu_{\mat{\Gamma}}\, \sqrt{\left(N-P\right)-1/\Delta-1}\sqrt{\left(P+1\right)-1/\Delta-1}},
	\label{explicit_psnr}
\end{equation}
with $R := \frac{\max_k\lvert a_k\rvert}{\min_k\lvert a_k\rvert}$, an explicit signal-to-noise threshold in which every quantity is known a priori except $\mu_{\mat{\Gamma}}$, computable from a single singular value decomposition of $\mat{G}\mat{\Gamma}^{-1/2}$, and the nontangentiality hypothesis, verifiable from the principal angles between $\mathbb{T}_P$ and $N_{\mathcal{H}_K}\left(\toeplitz{P}\left(\vec{x}^\star\right)\right)^\perp$, exactly $2K$ of which vanish at a nontangential model point, the cosine of the next one being the constant $c$ of Definition~\ref{def_nontangential}.
In every configuration tested in Section~\ref{section_exp}, including near-collision streams, this certificate confirms nontangentiality with $c$ bounded away from one, so that it certifies the fast inner rate rather than signalling degeneracy.
The only non-computable ingredient is the radius $r_0$ of Lemma~\ref{lemma_nontangential}: the hypothesis $r_0 \geq r_\delta/2$ of Corollary~\ref{cor_explicit} is not verifiable for a given instance, and a computable lower bound on $r_0$ in terms of $c\left(\mat{M}^\star\right)$ and the curvature of $\mathcal{E}$ is an open problem.
In the compressive regime $L < N$, the ball constant vanishes identically -- differences of ball points meet $\ker \mat{G}$ -- and injectivity can only hold in the \emph{secant-restricted} sense, over differences of points near the model set, of which the tangent constant~$\mu_T$ is the first-order instance; for randomized Fourier sampling, exact recovery from a number of measurements on the order of $K$ up to logarithmic factors is known to be possible under a comparable separation condition~\cite{tang2013}. Those guarantees concern the atomic-norm program; a high-probability lower bound on the secant-restricted constant under randomized Fourier sampling is, to our knowledge, open -- for sub-Gaussian measurements it follows from stable embeddings of the secant set of the model manifold -- and we state the Fourier case as a conjecture supported by those recovery results.

The specialization $\mat{G} = Id$ -- the vanilla FRI denoising problem, $\vec{y} = \vec{x}^\star + \vec{\epsilon}$ -- deserves a separate statement: it is the one setting in which the injectivity constant itself is closed-form, since $\mat{G}\mat{\Gamma}^{-1/2} = \mat{\Gamma}^{-1/2}$ is diagonal with entries $w_i^{-1/2}$.

\begin{corollary}[Vanilla FRI Denoising]
	\label{cor_vanilla}
	Consider the unconstrained denoising setting where $\mat{G} = \textnormal{Id}$ and $P \leq M$, so that $\gnormop{\mat{G}} = 1$, step size $\tau = 1/2$, and $\mu \geq \mu_{\mat{\Gamma}} = (P+1)^{-1/2}$.
	
	Under the hypotheses of Corollary~\ref{cor_explicit}, initializing Algorithm~\ref{algo_CPGD} at $\vec{x}_0 = \vec{y}$ yields the following guarantees with closed-form constants:
	\begin{enumerate}
		\item \emph{Iteration Complexity:} The contraction rate satisfies $q \leq \left(1 - (P+1)^{-1}\right)^{1/2}$, so GCPGD reaches tolerance $\textnormal{tol}$ in $O\left(P \ln\left(1/\textnormal{tol}\right)\right)$ iterations.
		
		\item \emph{Linear Convergence \& Error Floor:} The sequence converges linearly to a limit point $\bar{\vec{x}}$ satisfying
		\begin{equation}
			\gnorm{\bar{\vec{x}} - \vec{x}^\star} \leq 2\sqrt{P+1} \, \twonorm{\vec{\epsilon}},
		\end{equation}
		whenever the noise level obeys the explicit threshold
		\begin{equation}
			\label{vanilla_snr}
			\begin{aligned}
				\twonorm{\vec{\epsilon}} \leq{} & \frac{c_2 \left(\min_{1 \le k \le K} |a_k|\right)}{\sqrt{P+1}} \sqrt{(N-P) - \frac{1}{\Delta} - 1} \\
				& \times \sqrt{(P+1) - \frac{1}{\Delta} - 1},
			\end{aligned}
		\end{equation}
		with $c_2 := \frac{1-\delta}{2(2-\delta)}$.
	\end{enumerate}
\end{corollary}
\noindent\emph{The proof is provided in Section~\ref{sup_section_corr_2} in the supplementary material.}

For $P = M$, threshold~\eqref{vanilla_snr} simplifies to $c_2 \min_k |a_k| (M - 1/\Delta)/\sqrt{M+1}$, certifying that $P=M$ optimally maximizes noise tolerance at an $\mathcal{O}(\sqrt{M})$ rate. This yields two practical design rules: (i)~it establishes an exact stopping criterion $n^\star \approx \ln\left(\sigma_{K+1}(\toeplitz{P}(\vec{y}))/\textnormal{target}\right)/\ln(1/c^2)$ beyond which iterations overfit noise; and (ii)~since $r_\delta \propto \sigma_K$, the singular values of $\toeplitz{P}(\vec{y})$ enable adaptive model-order selection for $K$. Although the prefactor $2\sqrt{P+1}$ is conservative due to global $\mu_{\mat{\Gamma}}$, it provides a closed-form, input-free guarantee for Cadzow denoising.
It is interesting to note that CPGD, described in Algorithm~\num{1} in~\cite{simeoni2021}, corresponds to the choices $\alpha = 1$ and $\mat{\Gamma} = Id$ with a fixed step size $\tau$. The former is covered by Theorem~\ref{th_2}, whose conclusions hold at $\alpha = 1$: the relaxation is not needed for convergence. Lemma~\ref{lemma_confinement} shows, however, that it provably damps the denoiser drift -- the contraction factor $\left(1+\alpha\varepsilon\right)q$ and the noise amplification $\left(1+\alpha\varepsilon\right)b$ improve monotonically as $\alpha$ decreases -- while the model consistency of Remark~\ref{remark_feasibility} degrades as $1/\alpha$: the relaxation trades outer contraction against the feasibility of the limit. The latter measures the gradient in the unweighted Euclidean metric and reintroduces the conditioning factor $\mathrm{cond}\left(\mat{\Gamma}^{1/2}\right) \leq \sqrt{P+1}$ into the contraction constant $q$.
The provable improvement of GCPGD over CPGD is therefore driven by the metric, the explicit step size, and the drift damping of the relaxation.
Theorem~\ref{th_2} also sheds light on the empirical reliability of the CPGD family: within the basin, no random restarts are required, in contrast with the algorithm of Pan~\etal~\cite{pan2017} which relies on multiple random initializations, and the step size is computable a priori from a single power iteration on $\mat{G}\mat{\Gamma}^{-1/2}$.

\section{Experiments}
\label{section_exp}
The objectives of this Section are twofold: i) verify empirically different quantities predicted by the theory; and ii) assess the performance of GCPGD against CPGD and GenFRI. 
Hence, every experiment below, except in Section~\ref{subsec_comparison_exp}, overlays a quantity predicted by the theory on a quantity measured empirically, and the caption of each figure names the corresponding statement.
Throughout, the measurement operator is the irregular time-sampling operator $G_{l,m} = e^{j2\pi m \theta_l}$ with $\theta_l$ drawn uniformly in $\left[0, 1\right)$, and the noise is Gaussian with standard deviation $\sigma_n = \max_k \vert x_k \vert \exp\left(-\textnormal{PSNR}/10\right)$.
In the experiments of Sections~\ref{subsec_geometry_exp}--\ref{subsec_outer_exp}, GCPGD is run once from the least-squares warm start $\vec{x}_0 = \mat{G}\pinv\vec{y}$ of Corollary~\ref{cor_warm_start} with $L = 2N$ measurements, $P = M$, $n = 4$ inner Cadzow cycles, relaxation $\alpha = 1/2$, and the constant step size $\tau = \left(2\gnormop{\mat{G}}^2\right)^{-1}$, stopped at a relative step change of $10^{-12}$ or after $2500$ outer iterations.
Exact recovery is declared when $\gnorm{\hat{\vec{x}} - \vec{x}^\star} \leq 0.1\,\sigma_K$, with $\sigma_K$ that of the ground-truth signal. 
All predicted thresholds are evaluated at $\delta = 0.05$, i.e. $c_2 = \left(1-\delta\right)/\left(2\left(2-\delta\right)\right)$.
All the results are generated by the script \texttt{reproduce\_all\_experiments.py} provided in our Github repository~(\url{https://github.com/AdriBesson/gcpgd}).

\subsection{Geometry of the Model Set}
\label{subsec_geometry_exp}

\begin{figure}[htb]
	\centering
	\includegraphics[width=\columnwidth]{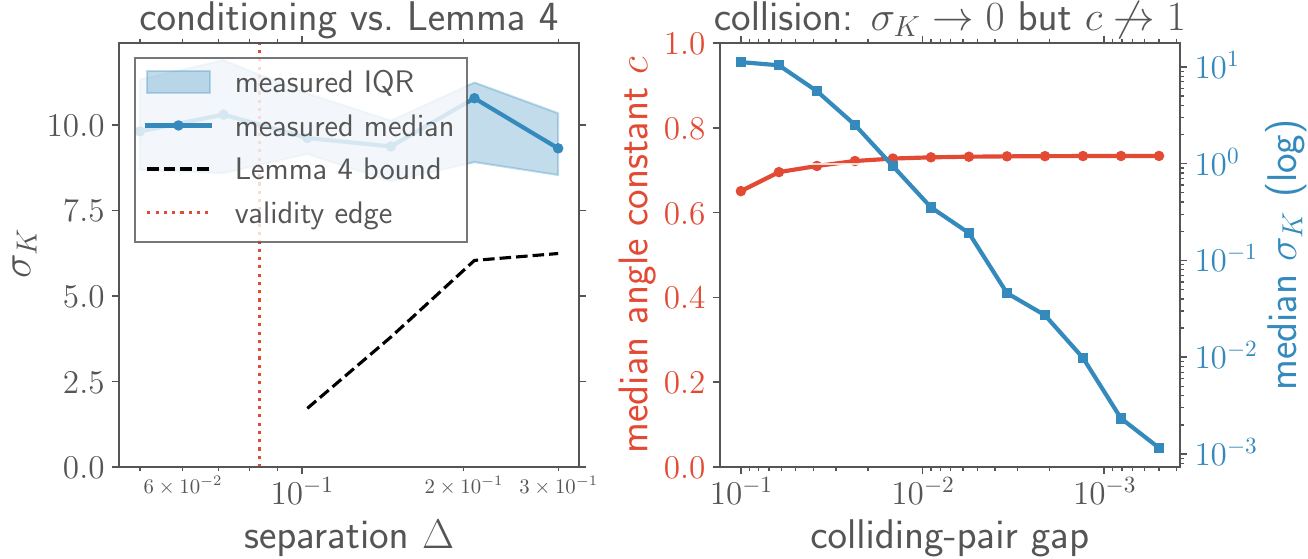}
	\caption{Geometry of the model set. Left: measured $\sigma_K$; the dotted vertical line marks the validity edge $\Delta = 1/\left(\min\left(N-P, P+1\right)-1\right)$. Right: collision sweep (one pair of Diracs merging, remaining Diracs fixed): $\sigma_K$ collapses over several orders of magnitude while the angle constant $c$ of Definition~\ref{def_nontangential} remains bounded away from one.}
	\label{fig_geometry}
\end{figure}

We first probe the two geometric quantities that govern the theory: the conditioning $\sigma_K$ and the nontangentiality certificate of Definition~\ref{def_nontangential}.
In the left panel of Figure~\ref{fig_geometry} (evaluated over \num{100} random streams per separation $\Delta \in [0.05, 0.24]$ across \num{14} logarithmically-spaced points, with $K=4$ and $M=20$), the bound of Lemma~\ref{lemma_conditioning} lower-bounds the measured $\sigma_K$ in every instance where it applies and tightens as the separation grows, expiring at the Rayleigh-type validity edge and holding in all valid instances.
In the right panel, driving the minimum separation to zero by merging a single pair of Diracs (evaluated over \num{12} logarithmically-spaced gaps down to $5 \times 10^{-4}$, with \num{25} trials per gap) collapses $\sigma_K$, while $c$ stays bounded away from one; moreover, exactly $2K$ principal angles vanish in all \num{300} instances of both ensembles.
Nontangentiality is thus generic and does not degrade at collisions: the basin of Theorem~\ref{th_1} closes through the conditioning $\sigma_K$ alone, as stated in Section~\ref{section_theory}.

Finally, we validate the a posteriori certificate of Remark~\ref{remark_certificate} (evaluated over \num{6} distinct $(K, M)$ configurations spanning $K \in \{2, 3, 4, 5\}$, $\Delta \in \{0.05, 0.08, 0.12, 0.20\}$, \num{20} cycles, and \num{25} trials per setup): for each configuration, the inner iteration is run from a perturbation of radius $r_\delta$ and the certified factor $1+\hat{\varepsilon}$ of~\eqref{certificate} is accumulated alongside the measured drift.
Every inner iterate remains in $\mathcal{G}_\delta$, the certified bound dominates the measured drift in \SI{100}{\percent} of the runs (with $\hat{\varepsilon}$ reliably bounded between \num{0.20} and \num{0.30}), and $\hat{\varepsilon}$ remains well below one at the theorem's own radius, so that the hypothesis of Corollary~\ref{cor_explicit} is certified rather than assumed throughout our experiments.

\subsection{Empirical Convergence of Alternating Projections}
\label{subsec_lip_const_exp}

Lemma~\ref{lemma_nontangential} makes a sharp quantitative prediction: the inner alternating projections contract the error toward their limit at the per-cycle rate $c^2$, with $c$ computable beforehand from a single singular value decomposition.
\begin{figure}[htb]
	\centering
	\includegraphics[width=0.7\columnwidth]{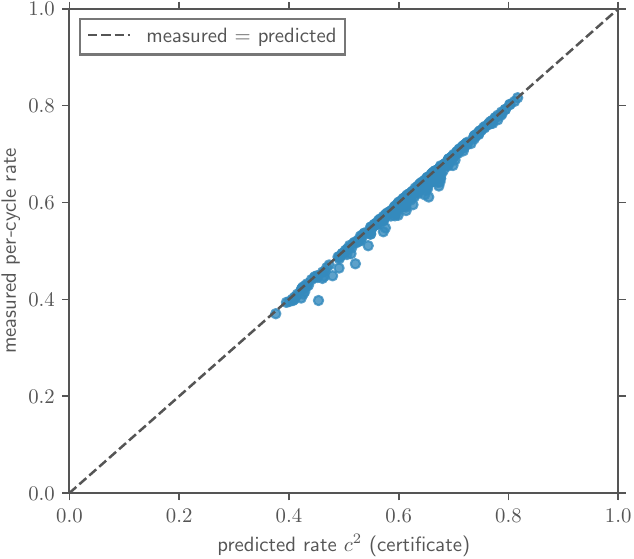}
	\caption{Certificate versus dynamics: measured per-cycle rate of the inner alternating projections against the prediction $c^2$ of Lemma~\ref{lemma_nontangential}, one point per configuration; the diagonal is the identity.}
	\label{fig_rate}
\end{figure}
Figure~\ref{fig_rate} compares this prediction with the rate measured by a log-linear fit of the iterates over an ensemble of configurations (evaluated over \num{20} trials per configuration for $K \in \{3, 5, 7, 9\}$, $M \in \{8, 10, 12, 18\}$, and $5$ separation distances $\Delta$): the points align with the diagonal, with a slight downward bias reflecting the transient of faster-decaying modes -- consistent with $\bar{c} \in \left(c^2, 1\right)$ being an asymptotic envelope, showing agreement up to three decimal places in the steady state.

\subsection{Convergence of GCPGD}
\label{subsec_outer_exp}

\begin{figure}[htb]
	\centering
	\includegraphics[width=\columnwidth]{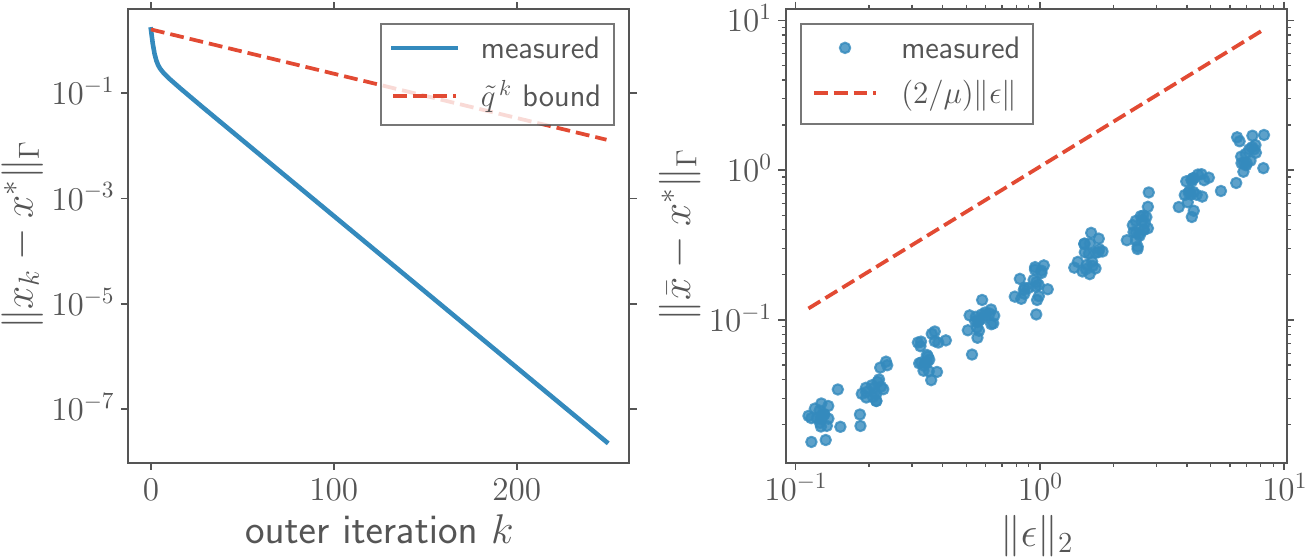}
	\caption{Convergence of GCPGD (Theorem~\ref{th_2}). Left: noiseless case, $\gnorm{\vec{x}_k - \vec{x}^\star}$ versus $k$; the measured linear rate lies below the predicted $\tilde{q}$ computed from the tangent-restricted $\mu_T$. Right: limiting error versus noise level: log-log slope one, prefactor below the bound $2/\mu_T$.}
	\label{fig_outer}
\end{figure}

Figure~\ref{fig_outer} validates the two core convergence claims of Theorem~\ref{th_2} (evaluated for $K=5$, $M=12$, $\Delta=0.10$ over \num{20} trials and up to \num{4000} iterations).

In the left panel of Figure~\ref{fig_outer}, we examine the noiseless regime ($\vec{\epsilon} = \vec{0}$) to evaluate the trajectory dynamics. The reconstruction error $\gnorm{\vec{x}_k - \vec{x}^\star}$ decays geometrically with the outer iteration index $k$. The measured linear contraction rate ($\approx \num{0.986}$) lies strictly below the theoretical prediction $\tilde{q} = \frac{1+q}{2}$ ($\approx \num{0.996}$) computed from the tangent-restricted constant $\mu_T$, confirming that the outer rate is reliably governed by $\mu_T$.

In the right panel of Figure~\ref{fig_outer}, we evaluate the noisy regime by plotting the limiting error $\gnorm{\bar{\vec{x}} - \vec{x}^\star}$ as a function of the noise level $\twonorm{\vec{\epsilon}}$ on a log-log scale. The empirical points align along a line of unit log-log slope, confirming the noise-proportional convergence floor predicted by Theorem~\ref{th_2}. Furthermore, the empirical prefactor ($\approx \num{0.21}$) remains well below the conservative stability bound $2/\mu_T$ ($\approx \num{2.36}$).
Because the global weighted singular value $\sigma_{\min}\left(\mat{G}\mat{\Gamma}^{-1/2}\right)$ is numerically zero in this setting, using the tangent-restricted constant $\mu_T$ is necessary rather than convenient.
Figure~\ref{fig_phase} tests Corollaries~\ref{cor_warm_start} and~\ref{cor_explicit} on the ensemble where the separation is binding (evaluated for $K=5$, $M=14$, with \num{40} trials per grid point, $\Delta$ spanning \num{10} logarithmically-spaced points from \num{0.005} to \num{0.24}, and PSNR ranging from \SI{-10}{\decibel} to \SI{55}{\decibel} in steps of \SI{3}{\decibel}): one pair of Diracs merges while the others stay fixed, so that $\sigma_K$ genuinely collapses and the recovery threshold slopes.
The amplitudes are unit-modulus with uniformly random phase, so that $R = 1$ and the threshold~\eqref{explicit_psnr} is governed by $\sigma_K$ alone; the calibration constant of the dotted curve is the median, over gaps, of the difference between the empirical \SI{50}{\percent}-crossing PSNR and the predicted threshold.
Three observations can be made.
First, the threshold predicted by Theorem~\ref{th_2} from the \emph{measured} medians of $\left(\sigma_K, \mu_T\right)$ runs parallel to the empirical transition across its entire \SI{40}{\decibel} rise; since $\mu_T$ is essentially constant over the sweep, this attributes the slope to $\sigma_K$ and validates the mechanism.
Second, a single order-one constant brings the predicted curve onto the transition: the \emph{shape} of the threshold is predicted by the theory (with the one-constant calibration offset by approximately \SI{5}{\decibel}).
Third, the fully explicit curve of Corollary~\ref{cor_explicit} certifies the well-separated region and expires at its validity edge before the collapse: the limiting factor of the closed-form guarantee is the Vandermonde bound of Lemma~\ref{lemma_conditioning}, not the mechanism.
The regime beyond the certified region is precisely the one in which learning-based reconstructions such as~\cite{leung2023} are motivated.

\begin{figure}[htb]
	\centering
	\includegraphics[width=\columnwidth]{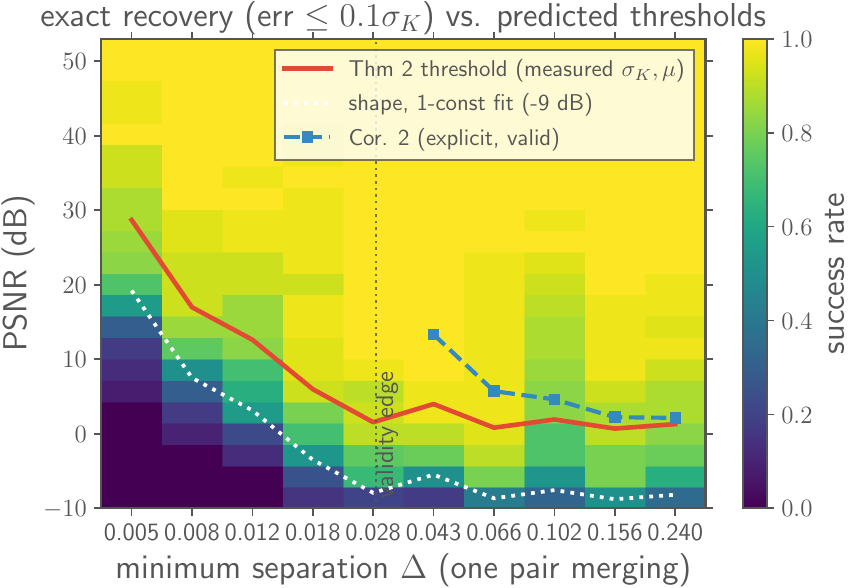}
	\caption{Exact-recovery rate on the collision ensemble (one pair of Diracs at gap $\Delta$, remaining Diracs fixed), with three predicted thresholds: the Theorem~\ref{th_2} threshold computed from the measured $\left(\sigma_K, \mu_T\right)$ (solid), its one-constant calibration (dotted), and the fully explicit Corollary~\ref{cor_explicit} on its validity region (dashed); the vertical line marks the validity edge.}
	\label{fig_phase}
\end{figure}

\subsection{Reconstruction Accuracy of GCPGD}
\label{subsec_comparison_exp}
\begin{figure*}[htb]
	\centering
	\subfloat[][$K=9, M=9, L=19$]{
		\includegraphics[width=0.32\textwidth]{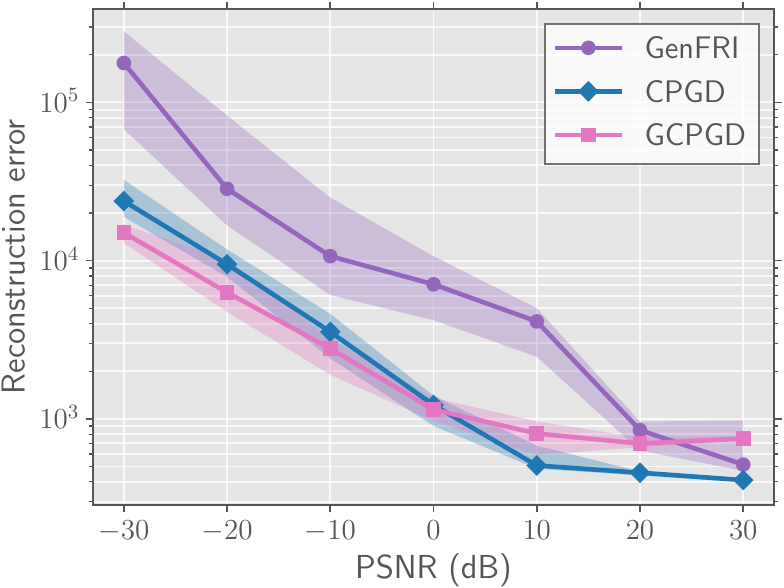}
		\label{accuracy_beta1}
	}
	\subfloat[][$K=9, M=18, L=19$]{
		\includegraphics[width=0.32\textwidth]{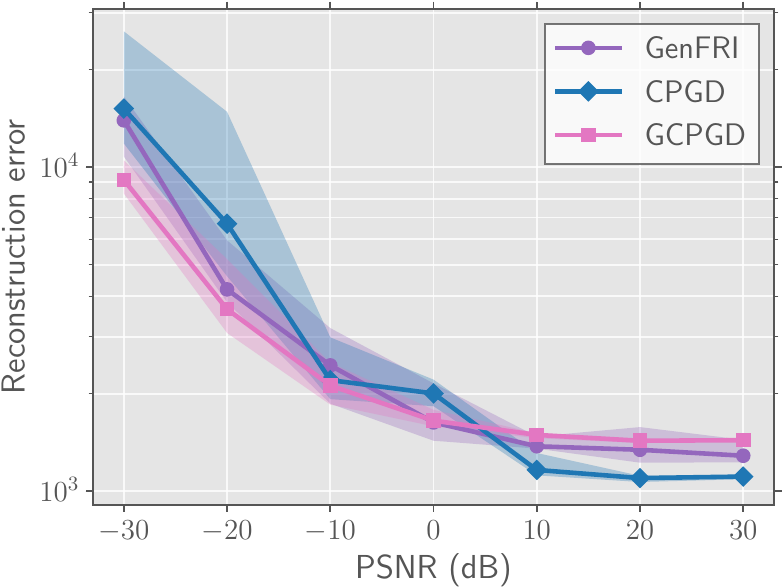}
		\label{accuracy_beta2}
	}
	\subfloat[][$K=9, M=27, L=19$]{
		\includegraphics[width=0.32\textwidth]{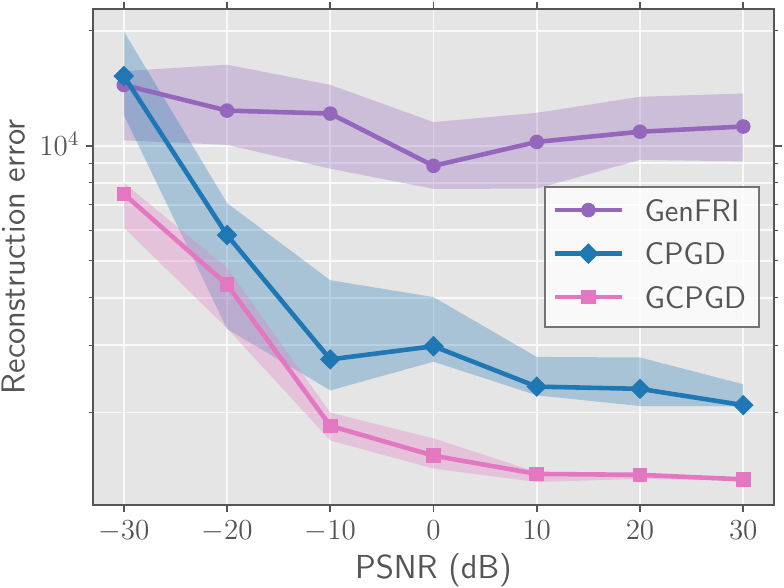}
		\label{accuracy_beta3}
	}
	\caption{Reconstruction error for the GenFRI algorithm~\cite{pan2017}, CPGD and GCPGD, various values of $M$, $P=M$, $K=9$, $L=19$ and a PSNR in $\{-30,-20,-10,0,10,20, 30\}$~\si{\decibel}.}
	\label{fig_reconstruction}
\end{figure*}
We compare numerically the reconstruction error of GCPGD against CPGD and the GenFRI algorithm~\cite{pan2017} in a similar way to~\cite{simeoni2021}.
We define a $1$-periodic stream of $K=9$ Diracs $x$ with amplitudes $x_k$ distributed with a log-normal distribution and locations $t_k$ drawn from a uniform distribution in $\left[0, 1\right]$ for $k = 1, \ldots, K$.
We generate $L=19$ noisy samples in the exact same way as in~\cite{simeoni2021}.
For both CPGD and GCPGD, the maximum number of outer iterations is set to $4000$. GCPGD terminates when the relative step change in the weighted metric satisfies $\gnorm{\vec{x}_{k+1} - \vec{x}_k} \le \textnormal{tol} \times \max\left(\gnorm{\vec{x}_k}, 10^{-12}\right)$ with $\textnormal{tol} = 10^{-7}$, employing the constant step size $\tau = \left(2\gnormop{\mat{G}}^2\right)^{-1}$. 
CPGD terminates when $\twonorm{\vec{x}_{k+1} - \vec{x}_k} / \twonorm{\vec{x}_k} < \textnormal{tol}$ with $\textnormal{tol} = 10^{-7}$, using the constant step size $\tau = \twonorm{\mat{G}}^{-2}$. In contrast, GenFRI uses an alternating-direction solver for structured total least squares without early stopping, running for a fixed $50$ iterations across $15$ random initializations to avoid local minima.
For the inner Cadzow denoising, we fix the number of iterations to~\num{5} for both CPGD and GCPGD, we choose $M=P$ and we apply no scaling to the denoiser, in accordance with the weighted-metric analysis of Section~\ref{section_theory}.
Figure~\ref{fig_reconstruction} displays the median and interquartile reconstruction error, \ie the $\ell_2-$norm between the recovered Fourier coefficients and the true ones, of the three algorithms as a function of the PSNR for $M \in \left\lbrace K, 2K, 3K \right\rbrace$.
The three panels traverse the regimes distinguished by the theory: at $M = 9$ the operator is square ($L = N = 19$), while at $M = 18$ and $M = 27$ the acquisition is compressive ($L < N$), the regime of Section~\ref{section_theory_cpgd} in which the ball constant vanishes and only restricted injectivity can hold.
Three observations can be made.
First, GCPGD attains the lowest error in nearly every configuration, with a visibly tighter interquartile band: within the basin of Theorem~\ref{th_2}, a single run from the least-squares warm start of Corollary~\ref{cor_warm_start} converges to the noise floor, whereas GenFRI requires $15$ random restarts and stagnates altogether at $M = 27$ -- at fixed $L$, enlarging the model no longer helps an unpreconditioned heuristic.
Second, the gap between GCPGD and CPGD widens with the oversampling: the two nearly coincide at high PSNR for $M = 2K$, while at $M = 3K$ CPGD plateaus strictly above GCPGD over the entire PSNR range.
This is the metric mechanism of Section~\ref{section_theory_cpgd} made visible: the unweighted gradient of CPGD pays the conditioning factor $\textnormal{cond}\left(\mat{\Gamma}^{1/2}\right) = \sqrt{P+1}$, which grows with $M$, while the $\mat{\Gamma}$-preconditioned step does not.
Third, the two ends of every curve match Theorem~\ref{th_2}: the degradation at low PSNR is consistent with the admissible-noise threshold proportional to $\sigma_K$, in agreement with the subspace-swap analysis of~\cite{wei2015guaranteed}, while the high-PSNR saturation is set by the stopping criterion and the iteration budget rather than by the noise -- within the basin the predicted error is linear in $\twonorm{\vec{\epsilon}}$, so the plateau reflects optimization accuracy, not a statistical limit.
The near-Rayleigh regime, in which the admissible noise of Corollary~\ref{cor_explicit} vanishes, delineates where learning-based reconstructions such as~\cite{leung2023} become the method of choice.

\section{Conclusion}
\label{section_conclusion}
We revisited generalized finite-rate-of-innovation reconstruction through the fixed-point set of the Cadzow denoiser. Two structural obstructions, the discontinuity of the rank projection across spectral-gap crossings and the non-convexity of the fixed-point set, exclude global guarantees for this family, and we developed a localized theory instead. We demonstrated that the Cadzow denoiser is quasi-nonexpansive on an explicit spectral tube whose radius scales with the conditioning $\sigma_K$ of the Dirac stream. 
We showed that the proposed GCPGD algorithm, equipped with $\mat{\Gamma}$-metric preconditioning, converges linearly within an explicit basin of attraction to a noise-proportional error floor -- to our knowledge, the first quantitative convergence guarantee in the GenFRI setting. 
The analysis separates the roles of its two ingredients: the denoiser confines the iterates to the rank variety, and the data term selects the model point within it. 
Apart from one localization radius, every hypothesis is computable per instance, through the nontangentiality certificate and the a posteriori drift factor $\hat{\varepsilon}$.
The recovery thresholds predicted from the measured $\left(\sigma_K, \mu_T\right)$ match the observed rates, floors, and phase transitions, and the explicit threshold of Corollary~\ref{cor_explicit}, expressed in separation and dynamic range, admits the subspace-swap condition of~\cite{wei2015guaranteed} as its average-case counterpart. 
As CPGD corresponds to a limiting parameter choice of GCPGD, the theory also explains its empirical reliability: a single run suffices within the basin, with no random restarts. 
Future direction will explore faster Cadzow implementations~\cite{wang2021} and include a systematic comparison with learning-based reconstruction~\cite{leung2023, kamath2025deepfri}.

\clearpage
\begin{center}
	\Large\bfseries Supplementary Material for ``Cadzow Projected Gradient Descent for Generalized Finite Rate of Innovation: A Quantitative Local Convergence Theory''
\end{center}
\renewcommand{\thesection}{S.\Roman{section}}
\renewcommand{\theequation}{S\arabic{equation}}
\setcounter{section}{0}
\setcounter{equation}{0}

\section{Proof of Proposition~\ref{prop_1} (Projection onto the Subset of rank-K Matrices)}
\label{sup_section_prop_1}
\begin{proof}
	If $\mat{X}$ has a rank lower than or equal to $K$, the statement is evident as $\Pi_{\mathcal{H}_K} \mat{X} = \mat{X}$ since $\mat{X} \in \mathcal{H}_K$.
	If $\mat{X}$ has a rank higher than $K$, the matrix $\hat{\mat{X}}$ given by the right-hand side of~\eqref{KSVD} attains $\min_{\mat{Y} \in \mathcal{H}_K}\fronorm{\mat{X} - \mat{Y}}$ by the Eckart--Young--Mirsky theorem~\cite{eckart1936}, so that the projection exists and equals $\hat{\mat{X}}$.
	Its uniqueness follows from the condition $\sigma_{K+1} \left(\mat{X}\right) < \sigma_{K} \left(\mat{X}\right)$ together with Proposition~\num{3.3} in~\cite{absil2012}.
\end{proof}

\section{Proof of Lemma~\ref{lemma_one_step} (One-Step Estimate)}
\label{sup_section_lemma_1}
\begin{proof}
	Uniqueness follows from Proposition~\ref{prop_1}, since $\rho\left(\mat{X}\right) < 1$ on $\mathcal{G}_\delta$.
	Write the singular value decomposition of $\mat{X}$ as $\mat{X} = \hat{\mat{X}} + \mat{N}$, where $\hat{\mat{X}} = \sum_{i \leq K}\sigma_i \vec{u}_i\vec{v}_i^\hermtransp$ retains the $K$ leading terms and the residual $\mat{N} = \sum_{i > K}\sigma_i \vec{u}_i\vec{v}_i^\hermtransp$ is supported on the orthogonal trailing singular directions, so that $\twonorm{\mat{N}} = \sigma_{K+1}\left(\mat{X}\right)$ and $\sigma_K\left(\hat{\mat{X}}\right) = \sigma_K\left(\mat{X}\right)$.
	We claim the exterior-sphere inequality for all $\mat{Y} \in \mathcal{H}_K$,
	\begin{equation}
		\textnormal{Re}\left\langle \mat{N}, \mat{Y} - \hat{\mat{X}}\right\rangle \leq \frac{\twonorm{\mat{N}}}{2\,\sigma_K\left(\mat{X}\right)}\fronorm{\mat{Y} - \hat{\mat{X}}}^2 = \frac{\rho\left(\mat{X}\right)}{2}\fronorm{\mat{Y} - \hat{\mat{X}}}^2, 
		\label{prox_reg}
	\end{equation}
	from two elementary facts.
	\emph{(i) Nearest-point expansion.} If $\hat{\vec{x}}$ is a nearest point of a closed set $S$ to $\vec{z}$, then for every $\vec{y} \in S$, $0 \leq \fronorm{\vec{z} - \vec{y}}^2 - \fronorm{\vec{z} - \hat{\vec{x}}}^2 = 2\,\textnormal{Re}\left\langle \vec{z} - \hat{\vec{x}}, \hat{\vec{x}} - \vec{y}\right\rangle + \fronorm{\hat{\vec{x}} - \vec{y}}^2$, \ie $\textnormal{Re}\left\langle \vec{z} - \hat{\vec{x}}, \vec{y} - \hat{\vec{x}}\right\rangle \leq \tfrac{1}{2}\fronorm{\vec{y} - \hat{\vec{x}}}^2$.
	\emph{(ii) Tail scaling.} For every $\lambda \in \left[1, \sigma_K\left(\mat{X}\right)/\sigma_{K+1}\left(\mat{X}\right)\right)$, the matrix $\mat{Z}_\lambda := \hat{\mat{X}} + \lambda\mat{N}$ has singular values $\left\lbrace \sigma_1, \ldots, \sigma_K\right\rbrace \cup \left\lbrace \lambda\sigma_{K+1}, \ldots \right\rbrace$, whose $K$ largest are still $\sigma_1, \ldots, \sigma_K$; by Eckart--Young--Mirsky, $\hat{\mat{X}}$ is a nearest point of $\mathcal{H}_K$ to $\mat{Z}_\lambda$.
	Applying (i) to the pair $\left(\mat{Z}_\lambda, \hat{\mat{X}}\right)$, whose normal is $\mat{Z}_\lambda - \hat{\mat{X}} = \lambda\mat{N}$, gives $\textnormal{Re}\left\langle \mat{N}, \mat{Y} - \hat{\mat{X}}\right\rangle \leq \left(2\lambda\right)^{-1}\fronorm{\mat{Y} - \hat{\mat{X}}}^2$ for every admissible $\lambda$; letting $\lambda = 1/\rho\left(\mat{X}\right)$ yields~\eqref{prox_reg}.
	The inequality is an instance of the prox-regularity of the rank variety established in~\cite{luke2013}; the derivation above makes the modulus explicit and self-contained, the spectral norm entering only through the admissible scaling range.
	Setting $\mat{Y} = \mat{M}^\star$ in the polarization identity
	\begin{multline*}
		\fronorm{\hat{\mat{X}} - \mat{M}^\star}^2 = \fronorm{\mat{X} - \mat{M}^\star}^2 - \fronorm{\mat{X} - \hat{\mat{X}}}^2 \\- 2\,\textnormal{Re}\left\langle \mat{X} - \hat{\mat{X}}, \hat{\mat{X}} - \mat{M}^\star \right\rangle
	\end{multline*}
	and using~\eqref{prox_reg} gives $\left(1 - \rho\left(\mat{X}\right)\right)\fronorm{\hat{\mat{X}} - \mat{M}^\star}^2 \leq \fronorm{\mat{X} - \mat{M}^\star}^2 - \fronorm{\mat{X} - \hat{\mat{X}}}^2$.
	Since $\mat{M}^\star \in \mathbb{T}_P$ and $\Pi_{\mathbb{T}_P}$ is the orthogonal projection onto a linear subspace, it is firmly nonexpansive, \ie $\fronorm{\Pi_{\mathbb{T}_P}\hat{\mat{X}} - \mat{M}^\star}^2 \leq \fronorm{\hat{\mat{X}} - \mat{M}^\star}^2 - \fronorm{\hat{\mat{X}} - \Pi_{\mathbb{T}_P}\hat{\mat{X}}}^2$, and combining the two displays yields~\eqref{one_step}.
\end{proof}

\section{Proof of Lemma~\ref{lemma_nontangential} (Local Convergence of the Inner Alternating Projections)}
\label{sup_section_lemma_2}
\begin{proof}
	The projection $\Pi_{\mathcal{H}_K}$ is single-valued near $\mat{M}^\star$, where the spectral gap is positive, by the prox-regularity of the rank variety~\cite{luke2013}.
	Write $u := \textnormal{dist}\left(\mat{M}_0, \mathbb{T}_P \cap \mathcal{H}_K\right)$ and $\Pi$ for the nearest-point projection onto $\mathcal{E}$.
	At a nontangential point, the alternating-projection theorem (Theorem~\num{6.1} of~\cite{andersson2013}) provides, for every $\varepsilon_1 > 0$, a radius on which: (T1)~$\mat{M}_k \rightarrow \mat{M}_\infty \in \mathbb{T}_P \cap \mathcal{H}_K$ with $\fronorm{\mat{M}_k - \mat{M}_\infty} \leq \bar{c}^{\,k} u$, the per-cycle rate $\bar{c} \in \left(c^2, 1\right)$ reflecting the two projections of one cycle, each contracting the sharpest surviving mode by the cosine $c$; and (T2)~the quasi-optimality of the limit, $\fronorm{\mat{M}_\infty - \Pi\left(\mat{M}_0\right)} \leq \varepsilon_1 u$.
	Equation~\eqref{ap_estimate_2} is immediate from (T1): $\textnormal{dist}\left(\mat{M}_k, \mathbb{T}_P \cap \mathcal{H}_K\right) \leq \fronorm{\mat{M}_k - \mat{M}_\infty} \leq \bar{c}^{\,k} u$.
	For~\eqref{ap_estimate_1}, note that the sequence started at $\mat{M}_k$ is the tail of the same sequence, with the same limit; the crude bound $\fronorm{\mat{M}_k - \mat{M}^\star} \leq \fronorm{\mat{M}_k - \mat{M}_\infty} + \fronorm{\mat{M}_\infty - \Pi\left(\mat{M}_0\right)} + \fronorm{\Pi\left(\mat{M}_0\right) - \mat{M}^\star} \leq 3\fronorm{\mat{M}_0 - \mat{M}^\star}$ keeps every iterate admissible after shrinking $r_0$ by the factor~$3$, so (T2) applied at $\mat{M}_k$ gives $\fronorm{\Pi\left(\mat{M}_k\right) - \Pi\left(\mat{M}_0\right)} \leq \fronorm{\Pi\left(\mat{M}_k\right) - \mat{M}_\infty} + \fronorm{\mat{M}_\infty - \Pi\left(\mat{M}_0\right)} \leq 2\varepsilon_1 u$.
	
	Since $\mathcal{E}$ is a $C^2$ manifold, $\mat{M}_j - \Pi\left(\mat{M}_j\right) \perp T_{\mathcal{E}}\left(\Pi\left(\mat{M}_j\right)\right)$ for every $j$, and chords of $\mathcal{E}$ deviate from tangency by at most a curvature constant $\kappa$ times their length; hence, with $v := \fronorm{\Pi\left(\mat{M}_0\right) - \mat{M}^\star}$,
	$\fronorm{\mat{M}_k - \mat{M}^\star}^2 \leq \left(1 + \varepsilon_2\right)\left[\bar{c}^{\,2k} u^2 + \left(v + 2\varepsilon_1 u\right)^2\right]$ and $\fronorm{\mat{M}_0 - \mat{M}^\star}^2 \geq \left(1 - \varepsilon_2\right)\left(u^2 + v^2\right)$ with $\varepsilon_2 = O\left(\kappa r_0 + \varepsilon_1\right)$; choosing $\varepsilon_1$ and $r_0$ small enough that $\left(1+\varepsilon_2\right)/\left(1-\varepsilon_2\right) \leq \left(1+\varepsilon\right)^2$ yields the first estimate.
	Finally, the results of~\cite{andersson2013} are stated for square Hankel-structured matrices; they transfer verbatim to the present Toeplitz lift through the row reversal $\mat{J}$, an orthogonal transformation that maps the Hankel subspace onto $\mathbb{T}_P$ and preserves ranks, the Frobenius geometry, and hence nontangentiality and all constants.
\end{proof}

\section{Proof of Theorem~\ref{th_1} (Local Regularity of Cadzow Denoising)}
\label{sup_section_th_1}
\begin{proof}
	Write $\mat{M}^\star := \toeplitz{P}\left(\vec{x}^\star\right)$.
	By the isometry property, $\fronorm{\mat{M}_0 - \mat{M}^\star} = \gnorm{\vec{x} - \vec{x}^\star} \leq r \leq r_\delta$, and Weyl's inequalities give $\sigma_K\left(\mat{M}_0\right) \geq \sigma_K - r_\delta > 0$ and $\sigma_{K+1}\left(\mat{M}_0\right) \leq r_\delta$, whence $\rho\left(\mat{M}_0\right) \leq \frac{r_\delta}{\sigma_K - r_\delta} = 1 - \delta$ by~\eqref{eq:tube_radius}, \ie $\mat{M}_0 \in \mathcal{G}_\delta$.
	By Lemma~\ref{lemma_nontangential}, $\fronorm{\mat{M}_k - \mat{M}^\star} \leq \left(1+\varepsilon\right)r \leq r_\delta$ for every $k$, so each $\mat{M}_k$ remains in $\mathcal{G}_\delta$ by the same Weyl argument, proving~(\num{1}), together with the estimates~\eqref{ap_estimate_1} and~\eqref{ap_estimate_2}.
	Since $\mat{M}_n \in \mathbb{T}_P$ and $\toeplitz{P}\left(H_n\left(\vec{x}\right)\right) = \mat{M}_n$, mapping~\eqref{ap_estimate_1} and~\eqref{ap_estimate_2} through the isometry gives~(\num{2}) and~(\num{3}).
\end{proof}

\section{Proof of Lemma~\ref{lemma_confinement} (Confinement)}
\label{sup_section_lemma_3}
\begin{proof}
	We proceed by induction, the case $k=0$ being the first condition in~\eqref{confinement_cond}.
	Assume $\rho_k := \gnorm{\vec{x}_k - \vec{x}^\star} \leq r$.
	Since $\vec{y} = \mat{G}\vec{x}^\star + \vec{\epsilon}$, the gradient step decomposes as
	\begin{equation*}
		\vec{v}_{k+1} - \vec{x}^\star = \left(Id - 2\tau\mat{A}\right)\left(\vec{x}_k - \vec{x}^\star\right) + 2\tau\mat{\Gamma}^{-1}\mat{G}^H\vec{\epsilon}, \; \mat{A} := \mat{\Gamma}^{-1}\mat{G}^H\mat{G}.
	\end{equation*}
	The operator $\mat{A}$ is self-adjoint and positive semidefinite for $\ginner{\cdot}{\cdot}$, with $\ginner{\mat{A}\vec{u}}{\vec{u}} = \twonorm{\mat{G}\vec{u}}^2$ and $\gnormop{\mat{A}} = \gnormop{\mat{G}}^2$, since $\mat{\Gamma}^{1/2}\mat{A}\mat{\Gamma}^{-1/2} = \mat{B}^\hermtransp\mat{B}$ with $\mat{B} := \mat{G}\mat{\Gamma}^{-1/2}$ and $\twonorm{\mat{B}^\hermtransp\mat{B}} = \twonorm{\mat{B}}^2 = \gnormop{\mat{G}}^2$.
	Writing $\vec{u} := \vec{x}_k - \vec{x}^\star$ and using $\gnorm{\mat{A}\vec{u}}^2 \leq \gnormop{\mat{G}}^2\ginner{\mat{A}\vec{u}}{\vec{u}} = \gnormop{\mat{G}}^2\twonorm{\mat{G}\vec{u}}^2$,
	\begin{align*}
		\gnorm{\left(Id - 2\tau\mat{A}\right)\vec{u}}^2 &= \gnorm{\vec{u}}^2 - 4\tau\twonorm{\mat{G}\vec{u}}^2 + 4\tau^2\gnorm{\mat{A}\vec{u}}^2 \\
		&\leq \gnorm{\vec{u}}^2 - 4\tau\left(1 - \tau\gnormop{\mat{G}}^2\right)\twonorm{\mat{G}\vec{u}}^2.
	\end{align*}
	For $\tau \leq \left(2\gnormop{\mat{G}}^2\right)^{-1}$ one has $1 - \tau\gnormop{\mat{G}}^2 \geq \tfrac12$, and $\twonorm{\mat{G}\vec{u}}^2 \geq \mu^2\gnorm{\vec{u}}^2$ since $\vec{u} \in \mathcal{B}_\delta\left(\vec{x}^\star\right) - \mathcal{B}_\delta\left(\vec{x}^\star\right)$; hence $\gnorm{\left(Id - 2\tau\mat{A}\right)\vec{u}} \leq q\gnorm{\vec{u}}$.
	The noise term satisfies $\gnorm{2\tau\mat{\Gamma}^{-1}\mat{G}^H\vec{\epsilon}} = 2\tau\twonorm{\mat{\Gamma}^{-1/2}\mat{G}^H\vec{\epsilon}} \leq 2\tau\gnormop{\mat{G}}\twonorm{\vec{\epsilon}} = b$.
	Therefore $\gnorm{\vec{v}_{k+1} - \vec{x}^\star} \leq q\rho_k + b \leq qr + b \leq r$ by~\eqref{confinement_cond}, so $\vec{v}_{k+1} \in \mathcal{B}_\delta\left(\vec{x}^\star\right)$ and~\eqref{eq:quasi_nonexpansive} gives $\gnorm{H_n\left(\vec{v}_{k+1}\right) - \vec{x}^\star} \leq \left(1+\varepsilon\right)\gnorm{\vec{v}_{k+1} - \vec{x}^\star}$.
	As $\vec{x}_{k+1} = \alpha H_n\left(\vec{v}_{k+1}\right) + \left(1-\alpha\right)\vec{v}_{k+1}$ is a convex combination,
	\begin{align*}
		\rho_{k+1} &\leq \alpha\gnorm{H_n\left(\vec{v}_{k+1}\right) - \vec{x}^\star} + \left(1-\alpha\right)\gnorm{\vec{v}_{k+1} - \vec{x}^\star} \\
		&\leq \left(\alpha\left(1+\varepsilon\right) + 1 - \alpha\right)\gnorm{\vec{v}_{k+1} - \vec{x}^\star} \\
		&\leq \left(1+\alpha\varepsilon\right)\left(q\rho_k + b\right) = \tilde{q}\rho_k + \left(1+\alpha\varepsilon\right)b.
	\end{align*}
	Unrolling this scalar recursion, $\rho_k \leq \tilde{q}^k\rho_0 + \left(1+\alpha\varepsilon\right)b\sum_{j=0}^{k-1}\tilde{q}^j \leq \tilde{q}^k\rho_0 + \tfrac{\left(1+\alpha\varepsilon\right)b}{1-\tilde{q}} \leq r$, which closes the induction and gives~\eqref{confinement_bound}.
\end{proof}

\section{Proof of Theorem~\ref{th_2} (Convergence of GCPGD)}
\label{sup_section_th_2}
\begin{proof}
	\emph{(1)} Apply Theorem~\ref{th_1} with $\varepsilon := \min\left(1, \frac{1-q}{2\alpha q}\right)$ and any $\bar{c} \in \left(c^2, 1\right)$, so that $\alpha\varepsilon \leq \frac{1-q}{2}$ and $\left(1+\alpha\varepsilon\right)q \leq \frac{1+q}{2} = \tilde{q} < 1$, and let $r = \min\left(r_\delta/\left(1+\varepsilon\right), r_0\right)$ be the corresponding radius.
	Set $c_2 := r/\sigma_K$ and $c_1 := \frac{\left(1-\tilde{q}\right)c_2}{2\left(1+\varepsilon\right)\tau\gnormop{\mat{G}}}$, so that the hypotheses give $\gnorm{\vec{x}_0 - \vec{x}^\star} \leq r$ and $\frac{\left(1+\alpha\varepsilon\right)b}{1-\tilde{q}} \leq r$ with $b = 2\tau\gnormop{\mat{G}}\twonorm{\vec{\epsilon}}$; both constants are strictly positive.
	Condition~\eqref{confinement_cond} of Lemma~\ref{lemma_confinement} therefore holds, giving $\vec{x}_k, \vec{v}_{k+1} \in \mathcal{B}_\delta\left(\vec{x}^\star\right)$ for all $k$; part~(i) of Theorem~\ref{th_1} places the inner Cadzow iterates in $\mathcal{G}_\delta$.
	\emph{(2)} The bound~\eqref{confinement_bound} of Lemma~\ref{lemma_confinement} reads $\gnorm{\vec{x}_k - \vec{x}^\star} \leq \tilde{q}^k\gnorm{\vec{x}_0 - \vec{x}^\star} + \frac{\left(1+\alpha\varepsilon\right)b}{1-\tilde{q}}$.
	Since $1 - \tilde{q} = \frac{1-q}{2} \geq \frac{\tau\mu^2}{2}$ and $1+\varepsilon \leq 2$, we have $\frac{\left(1+\alpha\varepsilon\right)b}{1-\tilde{q}} \leq \frac{8\tau\gnormop{\mat{G}}\twonorm{\vec{\epsilon}}}{\tau\mu^2} = \frac{8\gnormop{\mat{G}}}{\mu^2}\twonorm{\vec{\epsilon}}$, and letting $k \to \infty$ gives the $\limsup$.
	\emph{(3)} Let $\bar{\vec{x}}$ be as stated. Testing its optimality in Problem~\ref{redproloc} against the feasible point $\vec{x}^\star \in \Fix\left(H_n\right) \cap \mathcal{B}_\delta\left(\vec{x}^\star\right)$ gives $\twonorm{\mat{G}\bar{\vec{x}} - \vec{y}} \leq \twonorm{\mat{G}\vec{x}^\star - \vec{y}} = \twonorm{\vec{\epsilon}}$, hence $\twonorm{\mat{G}\left(\bar{\vec{x}} - \vec{x}^\star\right)} \leq 2\twonorm{\vec{\epsilon}}$, and the claim follows from the definition of the restricted injectivity constant $\mu$.
\end{proof}

\section{Proof of Lemma~\ref{lemma_conditioning} (Conditioning of the Dirac Stream)}
\label{sup_section_lemma_4}
\begin{proof}
	Writing $z_k := e^{-j2\pi t_k}$, the entries of $\toeplitz{P}\left(\vec{x}^\star\right)$ are the Fourier coefficients $\hat{x}_m = \sum_{k} a_k z_k^m$, so that $\toeplitz{P}\left(\vec{x}^\star\right) = \mat{V}_1 \mat{D} \mat{V}_2^T$ with $\left(\mat{V}_1\right)_{r,k} = z_k^r$, $\left(\mat{V}_2\right)_{c,k} = \bar{z}_k^{\,c}$ and $\mat{D} = \textnormal{diag}\left(a_k z_k^{P-M}\right)$: two Vandermonde matrices with unit-modulus nodes of identical wrap-around separation $\Delta$, of respective sizes $\left(N-P\right) \times K$ and $\left(P+1\right) \times K$.
	Since $\mat{V}_1$ has full column rank and $\mat{D}\mat{V}_2^T$ has full row rank, $\sigma_K\left(\mat{V}_1\mat{D}\mat{V}_2^T\right) \geq \sigma_K\left(\mat{V}_1\right) \min_k\lvert a_k\rvert\, \sigma_K\left(\mat{V}_2\right)$, and the bound $\sigma_K\left(\mat{V}\right)^2 \geq m - 1/\Delta - 1$, valid for an $m \times K$ Vandermonde matrix with unit-modulus nodes whenever $m > 1 + 1/\Delta$~\cite{moitra2015}, applied to $\mat{V}_1$ and $\mat{V}_2$, gives the result.
\end{proof}

\section{Proof of Corollary~\ref{cor_explicit} (Explicit Conditions)}
\label{sup_section_corr_1}
\begin{proof}
	Under $r_0 \geq r_\delta/2$ and $\varepsilon \leq 1$, the radius of Theorem~\ref{th_1} satisfies $r \geq r_\delta/2$, so that $c_2 = r/\sigma_K \geq \frac{1-\delta}{2\left(2-\delta\right)}$; the claim follows by inserting~\eqref{sigmaK_bound} into the signal-to-noise condition of Corollary~\ref{cor_warm_start}.
\end{proof}

\section{Proof of Corollary~\ref{cor_vanilla} (Vanilla FRI Denoising)}
\label{sup_section_corr_2}
\begin{proof}
	The singular values of $\mat{\Gamma}^{-1/2}$ are $w_i^{-1/2}$ with $\min_i w_i = 1$ and $\max_i w_i = \min\left(N-P, P+1\right) = P+1$ for $P \leq M$, giving $\gnormop{\mat{G}} = 1$ and $\mu_{\mat{\Gamma}} = \left(P+1\right)^{-1/2}$; insert these into Corollaries~\ref{cor_warm_start} and~\ref{cor_explicit}.
\end{proof}

\bibliographystyle{IEEEtran}
\bibliography{IEEEabrv, tsp_cpgd}
\end{document}